\documentclass[11pt,a4paper]{article}
\usepackage{fullpage}
\usepackage[english]{babel} 
\usepackage[T1]{fontenc} 
\usepackage{lmodern,microtype} 
\usepackage[font=large,labelfont={bf}]{caption}
\usepackage{amsmath,amsthm,amssymb,amsfonts}
\usepackage{dsfont,mathrsfs,ushort} 
\usepackage{titlesec,titling} 
\usepackage{geometry} 
\usepackage{setspace} 
\usepackage{float}
\usepackage{enumitem} 
\usepackage{booktabs}
\usepackage{multicol}
\usepackage{graphicx}
\usepackage{fancyvrb}
\usepackage{abstract}
\usepackage[numbers]{natbib}
\usepackage{pstricks} 
\usepackage{verbatim} 
\usepackage{hyperref}
\usepackage{xcolor}
\definecolor{winered}{rgb}{0.5,0,0}

\usepackage{setspace}
\usepackage{graphicx}
\usepackage{mathtools}
\usepackage{array}
\usepackage{tikz}
\usetikzlibrary{arrows, patterns}
\usetikzlibrary{arrows, patterns, positioning, decorations, decorations.markings}
\usepackage{multirow}
\usepackage{caption}
\usepackage{subcaption}
\usepackage{stackengine}
\usepackage[labelfont=bf,textfont=md]{caption}
\usepackage{algorithm}
\usepackage{float}
\usepackage{algorithmic}

\renewenvironment{proof}[1][Proof]{\noindent\textit{#1.} }{\ $\Box$\smallskip}

\usepackage{newtxmath}

\newcommand\xrowht[2][0]{\addstackgap[.5\dimexpr#2\relax]{\vphantom{#1}}}

\newcommand{\PreserveBackslash}[1]{\let\temp=\\#1\let\\=\temp}
\newcolumntype{C}[1]{>{\PreserveBackslash\centering}p{#1}}
\newcolumntype{R}[1]{>{\PreserveBackslash\raggedleft}p{#1}}
\newcolumntype{L}[1]{>{\PreserveBackslash\raggedright}p{#1}}

\newcommand\fnote[1]{\captionsetup{font=small}\caption*{#1}}

\hypersetup{
	pdfborder={0 0 0},
	colorlinks=false,
	linkcolor={winered},
	urlcolor={winered},
	filecolor={winered},
	citecolor={winered},
	linktoc=all,
}
\theoremstyle{plain}
\newtheorem{lemma}{Lemma}[section] %
\newtheorem{proposition}[lemma]{Proposition}%
\newtheorem{corollary}{Corollary}[section]
\theoremstyle{definition} %
\newtheorem{definition}{Definition}[section] %
\newtheorem{example}{Example}[section]%
\theoremstyle{remark}

\begin{document}
\title{Fairness Criteria for Allocating Indivisible Chores: \\ Connections and Efficiencies%
\thanks{A preliminary version of this paper appeared in \emph{Proceedings of AAMAS 2020}
 \citep{sunConnectionsFairnessCriteria2021a}.}}
\author{Ankang Sun\thanks{Warwick Business School, University of Warwick, Coventry, United Kingdom; A.Sun.2@warwick.ac.uk}
  \and Bo Chen\thanks{Corresponding author: Warwick Business School, University of Warwick, Coventry,
  United Kingdom; B.Chen@warwick.ac.uk}
  \and Xuan Vinh Doan\thanks{Warwick Business School,
  University of Warwick, Coventry, United Kingdom; Xuan.Doan@wbs.ac.uk} }
\date{}
\maketitle

\vspace{-3\bigskipamount}	
\begin{abstract}
We study several fairness notions in allocating indivisible \emph{chores} (i.e., items with
non-positive values) to agents who have additive and submodular cost functions. The fairness
criteria we are concern with are envy-free up to any item (EFX), envy-free up to one item (EF1),
maximin share (MMS), and pairwise maximin share (PMMS), which are proposed as relaxations of
envy-freeness in the setting of additive cost functions. For allocations under each fairness
criterion, we establish their approximation guarantee for other fairness criteria. Under the
additive setting, our results show strong connections between these fairness criteria and, at the
same time, reveal intrinsic differences between goods allocation and chores allocation. However,
such strong relationships cannot be inherited by the submodular setting, under which PMMS and MMS
are no longer relaxations of envy-freeness and, even worse, few non-trivial guarantees exist. We
also investigate efficiency loss under these fairness constraints and establish their prices of
fairness. 	

\end{abstract}

\section{Introduction}

Fair division is a central matter of concern in economics, multiagent systems, and artificial
intelligence \citep{bramsFairDivisionCakecutting1996, bouveretFairAllocationIndivisible2016a,
azizFairAssignmentIndivisible2015}. Over the years, there emerges a tremendous demand for fair
division when a set of indivisible resources, such as classrooms, tasks, and properties, are
divided among a group of agents. This field has attracted the attention of researchers and most
results are established when resources are considered as goods that bring positive utility to
agents. However, in real-life division problems, the resources to be allocated can also be chores
which, instead of positive utility, bring non-positive utility or cost to agents. For example, one
might need to assign tasks among workers, teaching load among teachers, sharing noxious facilities
among communities, and so on. Compared to goods, fairly dividing chores is relatively
under-developed. At first glance, dividing chores is similar to dividing goods. However, in
general, chores allocation is not covered by goods allocation and results established on goods do
not necessarily hold on chores. Existing works have already pointed out this difference in the
context of envy-freeness  \citep{bogomolnaiaDividingBadsAdditive2019,
bogomolnaiaCompetitiveDivisionMixed2017, branzeiAlgorithmsCompetitiveDivision2019} and equitability
\citep{freemanEquitableAllocationsIndivisible2019, freemanEquitableAllocationsIndivisible2020}. As
an example,  \citet{freemanEquitableAllocationsIndivisible2019} indicate that, when allocating
goods, a \emph{leximin}\footnote{A leximin solution selects the allocation that maximizes the
utility of the least well-off agent, subject to maximizing the utility of the second least, and so
on.} allocation is Pareto optimal and \emph{equitable up to any item}\footnote{Equitability
requires that any pair of agents are equally happy with their bundles. In equitability up to any
item allocations, the violation of equitability can be eliminated by removing any single item from
the happier (in goods allocation)/ less happy agent (in chores allocation).}, however, a leximin
solution does not guarantee equitability up to any item in chores allocation.

Among the variety of fairness notions in the literature, \emph{envy-freeness} (EF) is one of the
most compelling, which has drawn research attention over the past few decades
\cite{foleyResourceAllocationPublic1967,
bramsEnvyFreeCakeDivision1995,coleExistenceParetoEfficient2021}. In an envy-free allocation, no
agent envies another agent. Unfortunately, the existence of an envy-free allocation cannot be
guaranteed in general when the items are indivisible. A canonical example is {that one needs to
assign} one chore to two agents and the chore has a positive cost for either agent. Clearly, the
agent who receives the chore will envy the other. In addition, deciding the existence of an EF
allocation is computationally intractable, even for two agents with identical preference. Given this predicament, recent studies mainly
devote to relaxations of envy-freeness. One direct relaxation is known as \emph{envy-free up to one
item} (EF1) \cite{liptonApproximatelyFairAllocations2004a,
budishCombinatorialAssignmentProblem2011}. In an EF1 allocation, {one agent may be jealous of
another, but by removing one chore from the bundle of the envious agent, envy can be eliminated}. A
{similar but stricter notion} is \emph{envy-free up to any item} (EFX)
\cite{caragiannisUnreasonableFairnessMaximum2019}. In such an allocation, envy can be eliminated by
removing any positive-cost chore from the envious agent's bundle. Another fairness notion,
\emph{maximin share} (MMS) \cite{budishCombinatorialAssignmentProblem2011,
amanatidisApproximationAlgorithmsComputing2017}, generalizes the idea of ``cut-and-choose''
protocol in cake cutting. The maximin share is obtained by minimizing the maximum cost of a bundle
of an allocation over all allocations. The last fairness notion we consider is called
\emph{pairwise maximin share} (PMMS) \cite{caragiannisUnreasonableFairnessMaximum2019}, which is
similar to maximin share but different from MMS in that each agent partitions the combined bundle
of himself and any other agent into two bundles and then receives the one with the larger cost.

The existing research on envy-freeness and its relaxations concentrates on algorithmic features of
fairness criteria, such as their existence and (approximation) algorithms for finding them.
Relatively little research studies the connections between these fairness criteria themselves, or
the trade-off between these fairness criteria and the system efficiency, known as \emph{the price
of fairness}.

When allocating goods, \citet{amanatidisComparingApproximateRelaxations2018} compare the above four
relaxations of envy-freeness and provide results on the approximation guarantee of one to another.
However, these connections are unclear in allocating chores. On the price of fairness,
\citet{beiPriceFairnessIndivisible2019} study allocation of indivisible goods and focus on the
notions for which the corresponding fair allocations are guaranteed to exist, such as EF1, maximin
Nash welfare\footnote{Nash welfare is the product of agents' utilities.}, and leximin.
\citet{caragiannisEfficiencyFairDivision2012} study the price of fairness for both chores and
goods, and focus on the classical fairness notions, namely, EF, proportionality\footnote{An
allocation of goods (resp.\ chores) is proportional if the value (resp.\ cost) of every agent's
bundle is at least (resp.\ at most) one $n$-th fraction of his value (resp.\ cost) for all items.}
and equitability. To the best of our knowledge, no existing work covers the price of fairness for
the aforementioned four relaxations (which we will call \emph{additive relaxations} from time to
time in this paper) of envy-freeness in chores allocation.

In this paper, we fill these gaps by investigating the four relaxations of envy-freeness on two
aspects. On the one hand, we study the connections between {these} criteria and, in particular, we
consider the following questions: \emph{Does one fairness criterion implies another? To what extent
can one criterion guarantee for another?} On the other hand, we study the trade-off between
fairness and \emph{efficiency} (or \emph{social cost} defined as the sum of costs of the individual
agents). Specifically, for each fairness criterion, we investigate its \emph{price of fairness},
which is defined as the supremum ratio of the minimum social cost of a fair allocation to the
minimum social cost of any allocation.

\subsection{Main results}

On the connections between fairness criteria, we summarize our main results in
Figure~\ref{figure_connections} on the approximation guarantee of one fairness criterion for
another. Figures~\ref{Figure::additive} and \ref{Figure::submodular} show connections under
additive and submodular settings, respectively. As shown in Figure~\ref{Figure::additive} below,
when agents have additive cost functions, there exist evidently significant connections between
these fairness notions. While some of our results show similarity to those in goods allocation
\citep{amanatidisComparingApproximateRelaxations2018}, others also reveal their difference.
Figure~\ref{Figure::submodular} provides the corresponding results under the submodular setting,
which then show a sharp contrast to results under the additive setting. More specifically, except
that PMMS can have a bounded approximation guarantee on MMS, no non-trivial guarantee exists
between any other pair of fairness notions.

After comparing each pair of fairness notions, we compare the efficiency of fair allocations with
the optimal one. To quantify the efficiency loss, we apply the idea of the price of fairness and
our results are summarized in Table~\ref{table::price of fairness}. The terminology
``$\alpha$-XYZ'' below refers to an $\alpha$ approximation for fairness notion XYZ. The formal
definitions will be given in Section~\ref{sec:preliminaries}.

As detailed later in the paper, most of the results summarized in Figure~\ref{figure_connections}
and Table~\ref{table::price of fairness} are tight.

\begin{figure}
\captionsetup[subfigure]{justification=centering}
\vspace{-1.2\bigskipamount}
\begin{subfigure}[b]{.5\linewidth}
	\centering
	\begin{tikzpicture}[
		place/.style = {circle, minimum size=13mm}
		]
		\node (1) [place, draw] at (0,0)  {\scriptsize $\alpha$-EFX};
		\node (2) [place,draw ] at (6,0)  {\scriptsize $\alpha$-PMMS};
		\node (3) [place, draw] at (0,-8) {\scriptsize $\alpha$-MMS};
		\node (4) [place, draw] at (6, -8) {\scriptsize $\alpha$-EF1};

		\draw [>=latex, ->, thick] ([yshift= 7pt] 1.east) -- ([yshift= 7pt] 2.west) node [midway, above]
		{ \scriptsize $\textnormal{LB}= \textnormal{UB} = \frac{4\alpha}{2\alpha + 1 }$ (P\ref{prop::4.7})};
		
		\draw [>=latex, <-, thick]  ([yshift= -7pt] 1.east) -- ([yshift= -7pt] 2.west) node [midway,
		below] {\scriptsize $\alpha =1$: $\textnormal{UB=1}$ (P\ref{prop::5.1}), $\alpha >1$: $\textnormal{LB} =
			\infty$ (P\ref{prop::5.3})};
		
		\draw [>=latex, ->, thick]  ([xshift= 7pt] 2.south) -- ([xshift= 7pt] 4.north) node [midway,
		above, rotate=-90]{\scriptsize$\alpha =1$: $\textnormal{UB=1}$ (P\ref{coro::5.1}), $\alpha >1$:
			${\textnormal{LB}} = \infty $ (P\ref{prop::5.2})};

		\draw [>=latex, <-, thick]  ([xshift= -7pt] 2.south) -- ([xshift= -7pt] 4.north) node [midway,
		above,rotate=90] {\scriptsize $\textnormal{LB}= \textnormal{UB} = \frac{2\alpha + 1}{\alpha + 1 }$
			(P\ref{prop::4.8})};
		
		\draw [>=latex, ->, thick]  ([yshift= -7pt] 4.west) -- ([yshift= -7pt] 3.east) node [midway,
		below] {\scriptsize $\textnormal{LB} = \textnormal{UB} =
			\frac{n\alpha + n -1}{ n - 1 + \alpha } $ (P\ref{prop::4.6})};
		
		\draw [>=latex, <-, thick]  ([yshift= +7pt] 4.west) -- ([yshift= +7pt] 3.east) node [midway,
		above] {\scriptsize $\textnormal{LB} = \infty$ (P\ref{prop::5.7})};

		\draw [>=latex, ->, thick]  ([xshift= -7pt] 3.north) -- ([xshift= -7pt] 1.south) node [midway,
		above, rotate=90] { \scriptsize $\textnormal{LB} = \infty$ (P\ref{prop::5.7})};
		
		\draw [>=latex, <-, thick]  ([xshift= 7pt] 3.north) -- ([xshift= 7pt] 1.south) node [midway,
		above, rotate = -90] {\scriptsize $\textnormal{LB}$: $\max\{ \frac{2n\alpha}{2\alpha + 2n -3},
			\frac{2n}{n + 1} \}$,$\textnormal{UB}$: $\min\{ \frac{2n\alpha}{n-1+2\alpha},
			\frac{n\alpha + n -1}{n - 1+ \alpha} \} $ (P\ref{prop::4.5})};
		
		\draw [>=latex, ->, thick]  ([xshift= +7pt] 3.north) -- ([yshift= -7pt] 2.west) node [midway,
		above, rotate = 54.13010235] { \scriptsize $n\geq 3$: $\textnormal{LB} = 2$ (P\ref{prop::5.6})};
		\draw [>=latex, ->, thick]  ([xshift= -7pt] 2.south) -- ([yshift= +7pt] 3.east) node [midway,
		below, rotate = 54.13010235] { \scriptsize $\alpha<\frac{3}{2}$: $\textnormal{LB}=
			\frac{n\alpha}{\alpha + (n-1)(2-\alpha)}$, $\textnormal{UB} =
			\frac{n\alpha}{\alpha + (n-1) (1-\frac{\alpha}{2})}$ (P\ref{prop:4.9})};
	\end{tikzpicture}
	\subcaption{Additive}
	\label{Figure::additive}
\end{subfigure}
\begin{subfigure}[b]{.5\linewidth}
	\centering
	\begin{tikzpicture}[
		place/.style = {circle, minimum size=13mm}
		]
		\node (1) [place, draw] at (0,0)  {\scriptsize $\alpha$-EFX};
		\node (2) [place,draw ] at (6,0)  {\scriptsize $\alpha$-PMMS};
		\node (3) [place, draw] at (0,-8) {\scriptsize $\alpha$-MMS};
		\node (4) [place, draw] at (6, -8) {\scriptsize $\alpha$-EF1};

		\draw [>=latex, ->, thick] ([yshift= 7pt] 1.east) -- ([yshift= 7pt] 2.west) node [midway, above]
		{ \scriptsize $\textnormal{LB}= \textnormal{UB} = 2$ (P\ref{prop::sub_EFX_PMMS_MMS})};
		
		\draw [>=latex, <-, thick]  ([yshift= -7pt] 1.east) -- ([yshift= -7pt] 2.west) node [midway,
		below] {\scriptsize$\textnormal{LB} = \infty$ (P\ref{prop::6.3})};
		
		\draw [>=latex, ->, thick]  ([xshift= 7pt] 2.south) -- ([xshift= 7pt] 4.north) node [midway,
		above, rotate=-90]{\scriptsize $\textnormal{LB} = \infty$ (P\ref{prop::6.3})};

		\draw [>=latex, <-, thick]  ([xshift= -7pt] 2.south) -- ([xshift= -7pt] 4.north) node [midway,
		above,rotate=90] {\scriptsize $\textnormal{LB}= \textnormal{UB} =2$
			(P\ref{prop::sub_EF1_PMMS_MMS})};
		
		\draw [>=latex, ->, thick]  ([yshift= -7pt] 4.west) -- ([yshift= -7pt] 3.east) node [midway,
		below] {\scriptsize $\textnormal{LB} = \textnormal{UB} =
			n$ (P\ref{prop::sub_EF1_PMMS_MMS})};
		
		\draw [>=latex, <-, thick]  ([yshift= +7pt] 4.west) -- ([yshift= +7pt] 3.east) node [midway,
		above] {\scriptsize $\textnormal{LB} = \infty$ (P\ref{prop::sub_MMS_EFX_EF1})};

		\draw [>=latex, ->, thick]  ([xshift= -7pt] 3.north) -- ([xshift= -7pt] 1.south) node [midway,
		above, rotate=90] { \scriptsize $\textnormal{LB} = \infty$ (P\ref{prop::sub_MMS_EFX_EF1})};
		
		\draw [>=latex, <-, thick]  ([xshift= 7pt] 3.north) -- ([xshift= 7pt] 1.south) node [midway,
		above, rotate = -90] {\scriptsize $\textnormal{LB} =\textnormal{UB}= n$ (P\ref{prop::sub_EFX_PMMS_MMS})};
		
		\draw [>=latex, ->, thick]  ([xshift= +7pt] 3.north) -- ([yshift= -7pt] 2.west) node [midway,
		above, rotate = 54.13010235] { \scriptsize $n\geq 3$: $\textnormal{LB} = \textnormal{UB}= 2$ (P\ref{prop::sub_MMS_PMMS})};
		\draw [>=latex, ->, thick]  ([xshift= -7pt] 2.south) -- ([yshift= +7pt] 3.east) node [midway,
		below, rotate = 54.13010235] { \scriptsize $n \geq 3$: $\textnormal{LB} =\textnormal{UB} = \min\{n, \alpha\lceil \frac{n}{2}\rceil\}$ (P\ref{prop::sub_PMMS_MMS})};
	\end{tikzpicture}
	\subcaption{Submodular}
	\label{Figure::submodular}
\end{subfigure}
\fnote{\textnormal{Note: Figure~\ref{Figure::additive} and Figure~\ref{Figure::submodular} illustrate connections between fairness criteria under additive and submodular cost functions, respectively. LB and UB stand for lower and upper bound, respectively. P$x.y$ points to Proposition~$x$.$y$}}
\caption{Connections between fairness criteria}\label{figure_connections}
\end{figure}

\medskip

\begin{table}
\centering
\begin{tabular}{l| c c c c c |l}
\toprule \xrowht[()]{8pt}
	& EFX & EF1 & PMMS & $\frac{3}{2}$-PMMS & 2-MMS &\\
\midrule \xrowht[()]{9pt}
	\multirow{4}{*}{$n=2$}
	&\multicolumn{1}{c}{2} & \multicolumn{1}{c}{$\frac{5}{4}$} & \multicolumn{1}{c}{2}
    &  \multicolumn{1}{c|}{$\frac{7}{6}$}&  &{\multirow{2}[2]{*}{additive}} \\
	& \multicolumn{1}{c}{(P\ref{thm::price of PMMS and MMS})} &   \multicolumn{1}{c}{(P\ref{thm::7.1})}
    & \multicolumn{1}{c}{(P\ref{thm::price of PMMS and MMS})}
	& \multicolumn{1}{c|}{(P\ref{PoFn=2::PMMS})} & 1 &  \\
	\cline{2-5} \cline{7-7}\xrowht[()]{9pt}
	& \multicolumn{1}{c}{$[3,4)$} & \multicolumn{1}{c}{$[2,4)$} & \multicolumn{1}{c}{3}
    & \multicolumn{1}{c|}{$ [\frac{4}{3}, \frac{8}{3})$} & (L\ref{lemma::3.2}) & {\multirow{2}[2]{*}{submodular}}  \\
	&\multicolumn{1}{c}{(P\ref{sub_EFX})}	&\multicolumn{1}{c}{(P\ref{prop::pof_sub_EF1})}
    & \multicolumn{1}{c}{(P\ref{prop::pof_sub_MMS})} & \multicolumn{1}{c|}{(P\ref{prop::pof_sub_1.5PMMS})}
	&&\\
	\hline \xrowht[()]{8pt}
	\multirow{4}{*}{$ n \geq 3$}
	& & & & \multicolumn{1}{c|}{}& $[\frac{n+3}{6}, n )$ &{\multirow{2}[2]{*}{additive}}  \\
	& & & $\infty$
	&\multicolumn{1}{c|}{}&(P\ref{prop:my-add5})  & \\
	\cline{6-7} \xrowht[()]{8pt}
	& & & (P\ref{prop::7.3})
	& \multicolumn{1}{c|}{} & $[\frac{n+3}{6}, \frac{n^2}{2})$ & {\multirow{2}[2]{*}{submodular}} \\
	& & & & \multicolumn{1}{c|}{}&(P\ref{prop::pof_sub_2MMS_n}) & \\
\bottomrule
\end{tabular}
\fnote{\textnormal{Note: Interval $[a,b]$ means that the lower bound is equal to $a$ and upper bound is equal to $b$. P$x.y$ and L$x.y$ point to Proposition~$x$.$y$ and Lemma~$x$.$y$, respectively.}}
\caption{Prices of fairness}\label{table::price of fairness}
\end{table}

\subsection{Related works}

The fair division problem has been studied for both indivisible goods
\citep{liptonApproximatelyFairAllocations2004a, bezakovaivonaAllocatingIndivisibleGoods2005,
caragiannisUnreasonableFairnessMaximum2019} and indivisible chores
\citep{azizAlgorithmsMaxminShare2017, azizMaxminShareFair2019,
freemanEquitableAllocationsIndivisible2020}. Among various fairness notions, a prominent one is EF
proposed by \citet{foleyResourceAllocationPublic1967}. But an EF allocation may not exist and even
worse, checking the existence of an EF allocation is NP-complete
\citep{azizFairAssignmentIndivisible2015}. For the relaxations of envy-freeness, the notion of EF1
originates from \citet{liptonApproximatelyFairAllocations2004a} and is formally defined by
\citet{budishCombinatorialAssignmentProblem2011}. \citet{liptonApproximatelyFairAllocations2004a}
provide an efficient algorithm for EF1 allocations of goods when agents have monotone valuations
functions. When allocating chores, \citet{azizFairAllocationIndivisible2019} show that, in the
additive setting, EF1 is achievable by allocating chores in a round-robin fashion. Another fairness
notion that has been a subject of much attention in the last few years is MMS, proposed by
\citet{budishCombinatorialAssignmentProblem2011}. However, existence of an MMS allocation is not
guaranteed either for goods \citep{kurokawaFairEnoughGuaranteeing2018} or for chores
\citep{azizAlgorithmsMaxminShare2017}, even with additive valuation functions. Consequently, more
efforts are on approximation of MMS in the additive setting, with
\citet{amanatidisApproximationAlgorithmsComputing2017, ghodsiFairAllocationIndivisible2018,
gargImprovedApproximationAlgorithm2020} on goods and
\citet{azizAlgorithmsMaxminShare2017,huangAlgorithmicFrameworkApproximating2021} on chores. Some
other studies consider approximating MMS when agents have (a subclass of) submodular valuation
functions. \citet{barmanApproximationAlgorithmsMaximin2020} consider the submodular setting and
show that $0.21$-approximation of MMS can be efficiently computed by the round-robin algorithm.
\citet{barmanExistenceComputationMaximin2021} show that an MMS allocation is guaranteed to exist
and can be computed efficiently if agents have submodular valuation functions with binary margin.

The notions of EFX and PMMS are introduced by \citet{caragiannisUnreasonableFairnessMaximum2019}.
They consider goods allocation and establish that a PMMS allocation is also EFX when the valuation
functions are additive. Beyond the simple case of $n=2$, the existence of an EFX allocation has not
been settled in general. However, significant progress has been made for some special cases. When
$n=3$, the existence of an EFX allocation of goods is proved by
\citet{chaudhuryEFXExistsThree2020}. Based on a modified version of leximin {solutions},
\citet{plautAlmostEnvyFreenessGeneral2020} show that an EFX allocation is guaranteed to exist when
all agents have {identical} valuations. The work most related to ours is by
\citet{amanatidisComparingApproximateRelaxations2018}, which is on goods allocation under additive
setting, and provides connections between the above four relaxations of envy-freeness.

As for the price of fairness, {\citet{caragiannisEfficiencyFairDivision2012} show} that, in the
case of \emph{divisible} goods, the price of proportionality is $\Theta(\sqrt{n})$ and the price of
equitability is $\Theta(n)$. {\citet{bertsimasPriceFairness2011} extend} the study to other
fairness notions, \emph{maximin}\footnote{{It maximizes the lowest utility level among all the
agents.}} fairness and proportional fairness, and they provide a tight bound on the price of
fairness for a broad family of problems. {\citet{beiPriceFairnessIndivisible2019} focus} on
indivisible goods and concentrate on the fairness notions that are guaranteed to exist. They
present an asymptotically tight upper bound of $\Theta(n)$ on the price of maximum Nash welfare
\cite{coleApproximatingNashSocial2015}, maximum egalitarian welfare
\cite{bramsFairDivisionCakecutting1996} and leximin. They also consider the price of EF1 but leave
a gap between the upper bound $O(n)$ and lower bound $\Omega(\sqrt{n})$. This gap is later closed
by \citet{barmanSettlingPriceFairness2020} with the results that, for both EF1 and
$\frac{1}{2}$-MMS, the price of fairness is $O(\sqrt{n})$. All the work reviewed above on the price
of fairness is on the additive setting. On the other hand, the price of fairness has been studied
in other multi-agent systems, such as machine scheduling \citep{agnetisPriceFairnessTwoagent2019a}
and kidney exchange \citep{dickersonPriceFairnessKidney2014}.

\section{Preliminaries}
\label{sec:preliminaries}

In the problem of a fair division of indivisible chores, we have a set $N = \{1,2, \ldots,n\}$ of
agents and a set $E = \{e_1, \ldots, e_m \}$ of indivisible chores. As chores are the {items with
non-positive values}, each agent $i\in N$ is associated with a cost function $ c_i : 2^ {E}
\rightarrow R_{\geq 0}$, which maps any subset of $E$ into a non-negative real number. Throughout
this paper, we assume $c_i(\emptyset) = 0$ and $c_i$ is monotone, that is, $c_i (S ) \leq c_i(T)$
for any $S\subseteq T \subseteq E$. We say a (set) function $c(\cdot)$ is:
\begin{itemize}\vspace{-\medskipamount}
\item \emph{Additive}, if $c(S) = \sum _ { e \in S} c (e)$ for any $S \subseteq E$.
    \vspace{-\medskipamount}
\item \emph{Submodular},\footnote{An equivalent definition is as follows: $c(\cdot)$ is
    submodular if for any $S\subseteq T \subseteq E $ and $ e \in E \setminus T, c ( T \cup \{ e
    \}) - c (T) \leq c ( S \cup \{ e\}) - c(S)$.} if for any $S, T \subseteq E $, $c(S\cup T)+
    c(S\cap T) \leq c(S) + c(T)$. 
\item \emph{Subadditive}, if for any $S,T \subseteq E, c(S \cup T ) \leq c ( S ) + c ( T )$.
    \vspace{-\medskipamount}
\end{itemize}
Clearly, additivity implies submodularity, which in turn implies subadditivity. For simplicity,
instead of $c_i (\{ e_j \})$, we use $c_i (e_j)$ to represent the cost of chore $e_j$ for agent
$i$.

An \emph{allocation} $\mathbf{A}:=(A_1,\ldots,A_n)$ is an $n$-partition of $E$ among agents in $N$,
i.e., $A _ i \cap A _ j = \emptyset$ for any $i\neq j$ and $\cup _ {i\in N} A _ i = E$. Each subset
$S \subseteq E$ also refers to a \emph{bundle} of chores. For any bundle $S$ and $k \in
\mathbb{N}^+$, we denote by $\Pi _ k (S)$ the set of all $k$-partition of $S$, and $|S|$ the number
of chores in $S$.

\subsection{Fairness criteria}

We study envy-freeness and its (additive) relaxations and are concerned with both exact and approximate
versions of these fairness notions.

\begin{definition}
	For any $\alpha\ge 1$, an allocation $\mathbf{A} = ( A _1,\ldots,A _ n )$ is $\alpha$-EF if for any
	$i,j \in N, c _ i (A _ i ) \leq \alpha \cdot c _ i (A _ j)$. In particular, 1-EF is simply called
	EF.
\end{definition}

\begin{definition}
	For any $\alpha\ge 1$, an allocation $\mathbf{A}= ( A _1,\ldots,A _ n )$ is $\alpha$-EF1 if for any
	$ i , j \in N$, there exists $ e \in A _ i $ such that $ c _ i ( A _i \setminus \left\{ e\right\})
	\leq \alpha \cdot c _ i ( A_ j )$. In particular, 1-EF1 is simply called EF1.
\end{definition}

\begin{definition}
	For any $\alpha\ge 1$, an allocation $\mathbf{A}= ( A _1,\ldots,A _ n )$ is $\alpha$-EFX if for any
	$ i , j \in N, c _ i ( A _ i \setminus \left\{ e \right\}) \leq \alpha \cdot c _ i ( A _j )$ for
	any $e \in A _ i $ with $ c _ i (e ) > 0 $. In particular, 1-EFX is simply called EFX.
\end{definition}

Clearly, EFX\footnote{Note {\citet{plautAlmostEnvyFreenessGeneral2020}} consider a stronger version
of EFX by dropping the condition $c_i(e)>0$. In this paper, all results about EFX, except
Propositions~\ref{prop::5.1} and \ref{prop:my-add2}, still hold under the stronger version.} is
stricter than EF1. Next, {we formally introduce the notion of maximin share}. For any $k \in
[n]=\{1,\ldots, n\}$ and {bundle} $S \subseteq E$, the maximin share of agent $i$ on $S$ among $k$
agents is
\[
\textnormal{MMS}_i(k,S) = \min_{A \in \Pi _ k (S)} \max_{j \in [k] } c _ i ( A_ j ).
\]
{We are interested in the allocation in which} each agent receives cost no more than his maximin
share.

\begin{definition}
	For any $\alpha\ge 1$, an allocation $\mathbf{A} = ( A _1,\ldots,A _ n )$ is $\alpha$-MMS if for
	any $i\in N, c _ i ( A _ i ) \leq \alpha \cdot \textnormal{MMS} _ i (n, E)$. In particular, 1-MMS
	is called MMS.
\end{definition}

\begin{definition}\label{def2.5}
	For any $\alpha\ge 1$, an allocation $\mathbf{A}= ( A _1,\ldots,A _ n )$ is $\alpha$-PMMS if for
	any $i, j \in N$,
	\[
	c _ i ( A_ i )\leq \alpha \cdot \min_{\mathbf{B}\in \Pi_2 ( A_ i \cup A _ j )}
	\max \left\{ c _ i ( B _ 1), c _ i ( B _ 2)\right\}.
	\]
	In particular, 1-PMMS is called PMMS.
\end{definition}

Note that the right-hand side of the above inequality is equivalent to $\alpha \cdot
\textnormal{MMS}_ i (2, A _ i \cup A _ j )$.

\begin{example}
Let us consider an example with three agents and a set $E = \{ e_1, \ldots, e_7 \}$ of seven
chores. Agents have additive cost functions, displayed in the table below.
\begin{table}[H]
\centering
\begin{tabular}{c|ccccccc}
	\hline
	& $ e_ 1$ & $ e_ 2$ & $ e_ 3$ & $ e_ 4$ & $ e_ 5$ & $ e_ 6$ & $ e_ 7$  \\
	\hline
	Agent 1 & 2 & 3 & 3 & 0 & 4 & 2 & 1 \\
	Agent 2 & 3 & 1 & 3 & 2 & 5 & 0 & 5 \\
	Agent 3 & 1 & 5 & 10 & 2 & 3 & 1 & 3 \\
	\hline
\end{tabular}
\end{table}
It is not hard to verify that $\textnormal{MMS}_1 (3, E) = 5, \textnormal{MMS}_2( 3, E) = 7,$ $
\textnormal{MMS} _ 3 (3, E) = 10$. For instance, agent 2 can partition $E$ into three bundles: $\{e
_ 1, e_3\}, \{ e_ 2, e _ 7\},$ $\{ e _ 4, e _ 5, e _ 6 \}$, so that the maximum cost of any single
bundle for her is 7. Moreover, there is no other partitions that can guarantee a better worst-case
cost.

We now examine allocation $\mathbf{A} $ with $A_1 = \{ e _1, e _ 4, e _ 7 \}, A _ 2 =\{ e _ 2, e _
3, e _6\}, A _ 3 =\{ e _ 5\}$. We can verify that $c _ i(  A_ i ) \leq  c_ i ( A _ j )$ for any $ i
, j \in [3]$ and thus allocation $\mathbf{A}$ is \textnormal{EF} that is then also
\textnormal{EFX}, \textnormal{EF1}, \textnormal{MMS} and \textnormal{PMMS}. For another allocation
$\mathbf{B}$ with $B _1 =\{ e _ 1, e _5, e _ 7 \}, B _ 2$ $=\{ e _ 2, e _ 4, e _ 6 \}, B _ 3 =\{ e
_3\}$, agent 1 would still envy agent 2 even if chore $ e _ 7 $ is eliminated from her bundle, and
hence, allocation $\mathbf{B}$ is neither exact \textnormal{EF} nor \textnormal{EFX}. One can
verify that $\mathbf{B}$ is indeed $\frac{7}{3}$-\textnormal{EF} and $2$-\textnormal{EFX}.
Moreover, allocation $\mathbf{B}$ is \textnormal{EF1} because agent 1 would not envy others if
chore $ e _ 5$ is eliminated from her bundle and agent 3 would not envy others if chore $ e _ 3 $
is eliminated from her bundle. As for the approximation guarantee on the notions of MMS and PMMS,
it is not hard to verify that allocation $\mathbf{B}$ is $\frac{7}{5}$-\textnormal{MMS} and
$\frac{7}{5}$-\textnormal{PMMS}.
\end{example}

\subsection{Price of fairness}

Let $I = \langle N,E, (c_i)_{i\in N}\rangle$ be an instance of the problem for allocating
{indivisible chores} and let $\mathcal{I}$ be the set of all such instances. The \emph{social cost}
of an allocation $\mathbf{A} = (A_1,\ldots, A_n)$ is defined as $\textnormal{SC}(\mathbf{A}) =
\sum_{i \in N} c_i(A_i)$. The optimal social cost for an instance $I$, denoted by
$\textnormal{OPT}(I)$, is the minimum social cost over all allocations for this instance. Following
previous work \cite{caragiannisEfficiencyFairDivision2012,beiPriceFairnessIndivisible2019}, when
study the price of fairness, we assume that agents cost functions are normalized to one, i.e.,
$c_i(E) = 1$ for all $i\in N$.

The \emph{price of fairness} is the supremum ratio over all instances between the social cost of
the ``best'' fair allocation and the optimal social cost, where ``best'' refers to the one with the
minimum cost. Since we consider several fairness criteria, let $F$ be any given fairness criterion
and define by $F(I)$ as the set (possibly empty) of all allocations for instance $I$ that satisfy
fairness criterion $F$.

\begin{definition}
For any given fairness property $F$, the price of fairness with respect to $F$ is defined as
\begin{equation*}
	\textnormal{PoF} = \sup_{ I \in \mathcal{I} } \min_{\mathbf{A} \in F(I) }
	\frac{\textnormal{SC}(\mathbf{A})}{\textnormal{OPT}(I)}.
\end{equation*}
\end{definition}
Note that in the case where it is unclear whether $F(I)\neq\emptyset$ for any instances, we only
consider those instances $I$ with $F(I)\neq\emptyset$.

\subsection{Some simple observations}

We begin with some initial results, which reveal some intrinsic difference in allocating goods and
allocating chores as far as approximation guarantee is concerned. Our proof of any result in this
paper either immediately follows the statement of the result, or can be found in the Appendix
clearly indicated. First, we state a simple lemma concerning lower bounds of the maximin share.

\begin{lemma}\label{obs::3.1}
When agents have subadditive cost functions, for any $i\in N$ and $S \subseteq E$, we have
\[
\textnormal{MMS} _ i (k,S) \geq \frac{1}{k} c _ i (S), \forall k\in [n]; \qquad
\textnormal{MMS} _ i (k, S) \geq c _ i (e), \forall e \in S, \forall k\in [n].
\]
\end{lemma}

\begin{proof}
Let $\mathbf{T} = (T _1,\ldots,T_ k)$ be the $k$-partition of S defining $\textnormal{MMS} _ i (k ,
S)$; that is $\max_{T _ j} c _ i (T _ j ) = \textnormal{MMS}_ i (k, S)$. We start with the lower
bound $\frac{1}{k}c_i(S)$. Without loss of generality, assume $c_i(T_1) \geq c _i(T_2) \geq \cdots
\geq c_i(T_k)$ and as a result, we have $c_i(T_1) = \textnormal{MMS} _ i (k, S)$. Then, the
following holds
\begin{equation*}
	kc _ i (T _ 1) \geq \sum_{j=1}^{k} c _ i (T _ j ) \geq c _ i (\bigcup_{j=1}^{k} T _ j) = c _ i (S),
\end{equation*}
where the second transition is due to subadditivity. Due to $c_i(T_1) = \textnormal{MMS} _ i (k,
S)$, we have $\textnormal{MMS}_i (k, S) \geq \frac{1}{k} c _ i (S)$. As for the lower bound
$c_i(e)$, for any given chore $e \in S $, there must exist a bundle $T _ { j ^ {\prime} }$
containing $e$. Due to the monotonicity of cost function, we have $ c _ i (T _ { j^{\prime}}) \geq
c _ i (e)$, which combines $\textnormal{MMS} _ i ( k, S ) = c_1(T_1) \geq c_1(T_{j^{\prime}})$,
implying $ \textnormal{MMS} _ i ( k, S ) \geq c _ i (e)$.
\end{proof}

Based on the lower bounds in Lemma~\ref{obs::3.1}, we provide a trivial approximation guarantee for
PMMS and MMS.

\begin{lemma}\label{lemma::3.2}
When agents have subadditive cost functions, any allocation is 2-\textnormal{PMMS} and $n$-\textnormal{MMS}.
\end{lemma}
\begin{proof}
	Let $\mathbf{A} = (A _ 1, \ldots, A _ n)$ be an arbitrary allocation without any specified
	properties. We first show it's already an $n$-MMS allocation. By Lemma~\ref{obs::3.1}, for each agent $i$, we have $c _ i (E) \leq n\cdot\textnormal{MMS} _i(n,E)$. Then,
	due to the monotonicity of the cost function, $c_i(A_i) \leq c_i (E) \leq n \cdot
	\textnormal{MMS}_i(n, E)$ holds. 	
	
	Next, by a similar argument, we prove the result about 2-PMMS. By Lemma~\ref{obs::3.1}, $ c_i(A_i \cup A_j ) \leq 2 \textnormal{MMS}_i(2, A_i \cup A_j)$ holds for any $i,j
	\in N$. Again, due to the monotonicity of the cost function, we have $c_i(A_i) \leq c_i(A_i \cup
	A_j)$ that implies $c_i(A_i ) \leq 2 \textnormal{MMS}_i(2, A_i \cup A_j)$. Therefore, allocation
	$\mathbf{A}$ is also 2-PMMS, completing the proof.
\end{proof}

As can be seen from the proof of Lemma~\ref{lemma::3.2}, in allocating chores, if one assigns all
chores to one agent, then the allocation still has a bounded approximation for PMMS and MMS.
However, when allocating goods, if an agent receives nothing but his maximin share is positive,
then clearly the corresponding allocation has an infinite approximation guarantee for PMMS and MMS.

\section{Bounds on EF, EFX, and EF1 under additive setting}

Let us start with EF. According to the definitions, for any $\alpha \geq 1$, $\alpha$-EF is
stronger than $\alpha$-EFX and $\alpha$-EF1. The following propositions present the approximation
guarantee of $\alpha$-EF for MMS and PMMS.

\begin{proposition}\label{prop::4.1}
When agents have additive cost functions, for any $\alpha \geq 1$, an $\alpha$-\textnormal{EF}
allocation is also $\frac{n \alpha }{ n - 1 +\alpha}$-\textnormal{MMS}, and this result is tight.
\end{proposition}

\begin{proof}
We first prove the upper bound and focus on agent $i$. Let $\mathbf{A} = (A _ 1,\ldots,A _ n )$ be
an $\alpha$-EF allocation, then according to its definition, $c _ i (A_ i ) \leq \alpha \cdot c _ i
(A_ j )$ holds for any $j \in N$. By summing up $j$ over $ N\setminus\{ i \}$, we have $(n-1)c _
i(A_ i )\leq \alpha \cdot \sum_{ j \in N\setminus\{ i \}} c _ i (A _ j ) $ and as a result,
$(n-1+\alpha) c _ i (A_ i ) \leq  \alpha \cdot \sum_{ j \in N} c _ i (A _ j )  = \alpha \cdot c _ i
(E)$ where the last transition owing to the additivity of cost functions. On the other hand, from Lemma~\ref{obs::3.1}, it holds that $\textnormal{MMS} _ i (n ,E) \geq
\frac{1}{n}c _ i (E)$, implying the ratio $$\frac{c _ i (A _ i )}{\textnormal{MMS}_i (n,E)} \leq
\frac{n \alpha}{ n -1 + \alpha}.$$ 	

Regarding tightness, consider the following instance with $n$ agents and $n^2$ chores denoted as
$\{ e_1,\ldots,e_{n^2}\}$. Agents have an identical cost profile and for every $i \in [n]$, $ c_i (
e _ j )  = \alpha$ for $ 1 \leq  j \leq n  $ and $c _ i ( e_ j ) = 1$ for $ n+1  \leq j \leq n ^
2$. Consider allocation $\mathbf{B} = (B_1,\ldots,B_n)$ with $B_i = \{ e_{(i-1)n+1},\ldots,
e_{in}\}$ for any $i\in N$. It is not hard to verify that allocation $\mathbf{B}$ is $\alpha$-EF.
As for $\textnormal{MMS}_1(n,E)$, since in total we have $n$ chores with each cost $\alpha$ and
$(n-1)n$ chores with each cost 1, then in the partition defining $\textnormal{MMS}_1(n,E)$, each
bundle contains exactly one chore with cost $\alpha$ and $n-1$ chores with cost 1. Consequently, we
have $\textnormal{MMS}_1(n,E) = n-1+\alpha$ and the ratio $\frac{c_1(B_1)}{\textnormal{MMS}_1(n,E)
} = \frac{n\alpha}{n-1+\alpha}$.
\end{proof}

\begin{proposition}\label{prop::4.2}
When agents have additive cost functions, for any $\alpha \geq 1$, an $\alpha$-\textnormal{EF}
allocation is also $\frac{2\alpha}{1 +	\alpha}$-\textnormal{PMMS}, and this result is tight.
\end{proposition}

\begin{proof}
We first prove the upper bound. Let $\mathbf{A} = ( A _ 1, A _ 2, \ldots, A _ n)$ be an $\alpha$-EF
allocation, then according to the definition, for any $i,j \in N$, $c _ i ( A _ i )\leq
\alpha \cdot c _ i ( A _ j )$ holds. By additivity, we have $c _ i ( A _ i \cup A _ j ) = c _ i ( A
_ i ) + c _ i ( A _j ) \geq (1 + \frac{1}{\alpha}) \cdot c _ i ( A _ i )$, and consequently, $ c _
i( A _ i ) \leq \frac{\alpha}{ \alpha + 1 }\cdot c _ i ( A _ i \cup A _ j ) $ holds. On the other
hand, from Lemma \ref{obs::3.1}, we know $c _ i ( A _ i \cup A _ j ) \leq 2
\cdot \textnormal{MMS} _i (2, A _ i \cup A _ j)$, and therefore the following holds
\begin{equation*}
c _ i ( A _ i ) \leq \frac{2\alpha}{ \alpha + 1} \cdot\textnormal{MMS} _i (2, A _ i \cup A _ j).
\end{equation*}
As for tightness, consider an instance with $n$ agents and $2n$ chores denoted as $ \left\{e_1,
e_2,\ldots, e_{2n}\right\}$. Agents have identical cost profile and for every $  i \in [n]$, $c_i (
e _ 1 ) =  c _i ( e _ 2 ) =\alpha $ and $ c _ i ( e _ j ) =1 $ for $ 3 \leq j \leq 2n$. Now,
consider an allocation $\mathbf{B} = ( B_1,\ldots, B _n)$ where $B_i = \{ e_{2i-1}, e_{2i}\}$ for
any $i \in N$. It is not hard to verify that allocation $\mathbf{B}$ is $\alpha$-EF and except for
agent 1, no one else will violate the condition of PMMS. For any $j \geq 2$, one can calculate
$\textnormal{MMS}_1(2, B_1\cup B_j) = 1 + \alpha$, yielding the ratio
$\frac{c_1(B_1)}{\textnormal{MMS}_1(2, B_1\cup B_j)} = \frac{2\alpha}{1+\alpha}$, as required.
\end{proof}

Proposition~\ref{prop::4.2} indicates that the approximation guarantee of $\alpha$-EF for PMMS is
independent of the number of agents. However, according to Proposition~\ref{prop::4.1}, its
approximation guarantee for MMS is affected by the number of agents. Moreover, this guarantee ratio
converges to $\alpha$ as $n$ goes to infinity.

We remark that none of EFX, EF1, PMMS and MMS has a bounded guarantee for EF. We show this by a
simple example. Consider an instance of two agents and one chore, and the chore has a positive cost
for both agents. Assigning the chore to an arbitrary agent results in an allocation that satisfies
EFX, EF1, PMMS and MMS, simultaneously. However, since one agent has a positive cost on his own
bundle and zero cost on other agents' bundle, such an allocation has an infinite approximation
guarantee for EF.

Next, we consider approximation of EFX and EF1.
\begin{proposition}\label{prop::4.3}
When agents have additive cost functions, an $\alpha$-\textnormal{EFX} allocation is
$\alpha$-\textnormal{EF1} for any $\alpha \geq 1$. On the other hand, an \textnormal{EF1}
allocation is not $\beta$-\textnormal{EFX} for any $\beta \geq 1$.
\end{proposition}

\begin{proof}
We first show the positive part. Let $\mathbf{A} = ( A _ 1 , A _ 2,\ldots, A _ n )$ be an
$\alpha$-EFX allocation, then according to its definition, $\forall i ,j \in N, \forall e \in A _
i$ with $c_i(e) >0$, $c _ i ( A _ i \setminus \left\{ e \right\}) \leq \alpha \cdot c _ i ( A _ j )$ holds. This
implies $\mathbf{A}$ is also $\alpha$-EF1. 	

For the impossibility result, consider an instance with $n$ agents and $2n$ chores denoted as  $\{
e_1,e_2,\ldots,$ $e_{2n}\}$. Agents have identical cost profile. The cost function of agent 1
is: $c_1(e_1) = p, c_1(e_j)=1, \forall j \geq 2$ where $p\gg 1$. Now consider an allocation
$\mathbf{B} = (B_1,\ldots,B_n)$ with $B_i = \{ e_{2i-1}, e_{2i}\}, \forall i \in N$. It is not hard
to see allocation $\mathbf{B}$ is EF1 and except for agent 1, no one else will envy the bundle of
others. Thus, we only concern agent 1 when calculate the approximation guarantee for EFX. By
removing chore $e_2$ from bundle $B_1$, $\frac{c_1(B_1\setminus \{ e_2\})}{c_1(B_j)} = \frac{p}{2}$
holds for any $j \in N\setminus\{ 1\}$, and the ratio $\frac{p}{2} \rightarrow \infty$ as $p
\rightarrow \infty$.
\end{proof}

Next, we consider the approximation guarantee of EF1 for MMS. In allocating goods, {\citet{amanatidisComparingApproximateRelaxations2018} present} a tight result that an $\alpha$-EF1
allocation is $O(n)$-MMS. In contrast, in allocating chores, $\alpha$-EF1 can have a much better
guarantee for MMS.

\begin{proposition}\label{prop::4.6}
When agents have additive cost functions, for any $\alpha \geq 1$ and $n \geq 2$, an
$\alpha$-\textnormal{EF1} allocation is also $\frac{n\alpha + n -1}{ n -1 +
\alpha}$-\textnormal{MMS}, and this result is tight.
\end{proposition}

\begin{proof}
We first prove the upper bound.  Let $\mathbf{A}= ( A _ 1,\ldots, A _ n)$ be an $\alpha$-EF1
allocation and the approximation guarantee for MMS is determined by agent $i$. We can further assume $c_i(A_i)>0$; otherwise agent $i$ meets the condition of MMS and we are done. Let {$\bar{e}$} be the chore with largest cost for
agent $i$ in bundle $A _ i $, i.e., ${\bar{e}} \in \arg\max_{e \in A _ i } c _ i (e)$.

By the definition of $\alpha$-EF1, for any $ j \in N\setminus \{ i \}$, $c _ i (A _ i \setminus \{{\bar{e}}\}) \leq \alpha \cdot c _ i (A _ j )$ holds. Then, by summing up $j$ over $N\setminus \{ i \}$
and adding a term $\alpha c _ i (A _ i )$ on both sides, the following holds,
\begin{equation}\label{eq::9}
	\alpha \cdot \sum _ { j \in N} c _ i (A _ j ) \geq (n - 1 + \alpha ) c _ i (A _ i ) - ( n -1) c _ i ({\bar{e}}).
\end{equation}
From Lemma~\ref{obs::3.1}, we have $\textnormal{MMS} _ i (n,E) \geq \max \{ \frac{1}{n}c _ i (E), c
_ i ({\bar{e}}) \}$, and by additivity, it holds that
\begin{equation}\label{eq::10}
	n\alpha \textnormal{MMS}_i(n, E) \geq (n-1+\alpha) c_i(A_i) - (n-1) \textnormal{MMS}_i(n,E).
\end{equation}
Inequality (\ref{eq::10}) is equivalent to $\frac{c _ i (A_ i )}{ \textnormal{MMS} _ i (n,M)} \leq
\frac{n\alpha + n -1}{ n -1 + \alpha}$, as required.

As for tightness, consider the following instance with $n$ agents and a set
$E=\{e_1,\ldots,e_{n^2-n+1}\}$ of $n^2-n+1$ chores. Agents have an identical cost profile and for
every $ i \in [n]$, $c _ i ( e _ 1 ) = \alpha +  n - 1$, $c _ i ( e _ j ) = \alpha$ for any $2 \leq
j \leq n $ and $ c _ i ( e  _ j ) = 1 $ for $  j \geq  n + 1 $. Now, consider an allocation
$\mathbf{B} = \left\{ B _ 1,\ldots, B_  n \right\}$ with $B _ 1 =\left\{ e _1, \ldots, e _ n
\right\}$ and $B_j =\{ e_{n+(n-1)(j-2)+1}, \ldots,$ $e_{n+(n-1)(j-1)}\}$ for any $j\geq 2$. Then,
we have $c _i(B_j ) = n-1$ for any $ i \in [n]$ and $  j \geq 2$. Accordingly, except for agent 1,
no one else will violate the condition of $\alpha$-EF1 and MMS. As for agent 1, since
$c_1(B_1\setminus\{e_1\}) = (n-1)\alpha = \alpha c_1(B_j), \forall j\geq 2$, then we can claim that
allocation $\mathbf{B}$ is $\alpha$-EF1. To calculate $\textnormal{MMS}_1(n,E)$, consider an
allocation $\mathbf{T} = ( T _1 ,\ldots, T _ n )$ with $T_1 = \{ e_1\}$ and $T_ j = \{ B _ j \cup
\left\{ e _ j \right\} \}$ for any  $ 2\leq j \leq n $. It is not hard to verify that $c_1(T_j) =
\alpha + n -1$ for any $j\in N$. Therefore, we have $\textnormal{MMS}_1(n,E) = \alpha + n -1 $
implying the ratio $\frac{c _ 1 (B _ 1 )}{ \textnormal{MMS} _ 1 (n,E)} = \frac{n\alpha + n -1}{ n
-1 + \alpha}$, completing the proof.
\end{proof}

We now study $\alpha$-EFX in terms of its approximation guarantee for MMS and provide upper and
lower bounds for general $\alpha \geq 1$ or $n \geq 2$.

\begin{proposition}\label{prop::4.5}
When agents have additive cost functions, for any $\alpha \geq 1$ and $n\geq  2$, an
$\alpha$-\textnormal{EFX} allocation is $ \min\left\{ \frac{2n\alpha}{n-1+2\alpha}, \frac{n\alpha +
n -1}{n - 1+ \alpha} \right\}$-\textnormal{MMS}, while it is not $\beta$-\textnormal{MMS} for any
$\beta < \max\left\{ \frac{2n\alpha}{2\alpha + 2n -3}, \frac{2n}{n + 1}\right \}$.
\end{proposition}

\begin{proof}
We first prove the upper bound. Let $\mathbf{A} = ( A _1,\ldots, A _ n )$ be an $\alpha$-EFX
allocation with $\alpha \geq 1$ and the approximation guarantee for MMS is determined by agent $i$. The upper bound $\frac{n\alpha + n -1 }{
n - 1 + \alpha}$ directly follows from Proposition \ref{prop::4.3} and \ref{prop::4.6}. In what
follows, we prove the upper bound $\frac{2n\alpha}{n-1+2\alpha}$. We assume $c_i(A_i)>0$; otherwise agent $i$ meets the condition of MMS and we are done. We further assume that every chore in $A_i$ has positive cost for agent $i$ since zero-cost chore does not affect the approximation guarantee for EFX or MMS. Let
$e^{*} $ be the chore in bundle $A_i$ having the minimum cost for agent $i$, i.e., $e^{*} \in \arg
\min _ {e \in A _ i } c _ i (e)$. Next, we divide the proof into two cases. 	

\emph{Case 1}: $ |A _ i | = 1$. Then $ e^{*}$ is the unique element in $A _ i $, and thus $ c _ i (
A _ i ) = c _ i ( e ^ {*})$. By Lemma \ref{obs::3.1}, $c_i(e^*) \leq
\textnormal{MMS}_i(n, E)$ holds, and thus, $ c _ i ( A _ i ) \leq \textnormal{MMS}_i(n, E) $. 	

\emph{Case 2}: $ |A _ i | \geq 2$. By the definition of $\alpha$-EFX, for any $j \in N
\setminus \left\{ i\right\}$, $c _ i ( A _ i \setminus \left\{ e ^ {*} \right\}) \leq \alpha \cdot c
_ i ( A _ j )$. Since $e ^ {*} \in \arg \min _ {e \in A _ i } c _ i (e)$ and $|A_i| \geq 2$, we
have $ c _ i (e ^ {*}) \leq \frac{1}{2}c _ i ( A _ i )$. Then, the following holds,
\begin{equation}\label{eq::6}
	\alpha \cdot c _ i ( A _ j ) \geq c _ i ( A _ i) - c _ i (e^{*}) \geq
	\frac{1}{2} c _ i ( A _ i ), \qquad \forall j \in N\setminus \left\{ i \right\}.
\end{equation}
By summing up $j$ over $ N \setminus \left\{ i \right\}$ and adding a term $\alpha c _ i ( A_ i )$
on both sides of inequality (\ref{eq::6}), the following holds
\begin{equation}\label{eq::7}
	\alpha \cdot c _ i ( E ) = \alpha\cdot \sum_{ j \in N\setminus \left\{i\right\} } c _ i ( A _ j )
	+ \alpha \cdot c _ i ( A _ i ) \geq \frac{n-1 + 2\alpha}{2}c _ i ( A _i ).
\end{equation}
On the other hand, from Lemma~\ref{obs::3.1}, we know $\textnormal{MMS} _ i (n ,
E) \geq \frac{1}{n} c _ i ( E)$, which combines inequality (\ref{eq::7}) yielding the ratio
\begin{equation*}
	\frac{c _ i ( A _ i )}{ \textnormal{MMS} _ i (n , M)} \leq \frac{2n\alpha}{n - 1 + 2\alpha}.
\end{equation*}
Regarding the lower bound $\frac{2n}{n +1 }$, consider an instance with $n$ agents and a set $E =
\left\{ e_1 , e _ 2 ,..., e_ {2n}\right\}$ of $2n$ chores. Agents have identical cost profile and $
c  _i  ( e _ j ) =\lceil \frac{j}{2} \rceil $ for any $i,j$. It is not hard to verify that for any
$i\in [n]$, $\textnormal{MMS}_i(n, E) = n+1$. Then, consider the allocation $\mathbf{B} = ( B_ 1,
..., B_n)$ with $B _ 1 = \left\{ e _ {2n -1}, e _ {2n}\right\}$ and $B _ i = \{ e _ {i -1}, e _ {2n
- i}\}$ for any $i \geq 2$. Accordingly, we have $ c  _ i( B _ j ) = n $ for any $  i \in [n]$ and $
j \geq 2$. Thus, except for agent 1, no one else will violate the condition of MMS and EFX. As for
agent 1, since $c_1(B_1\setminus\{e_{2n}\}) = c_1(B_1\setminus\{e_{2n-1}\}) = n $, envy can be
eliminated by removing any single chore . Hence, the allocation $\mathbf{B}$ is EFX and its
approximation guarantee for MMS equals to $\frac{c_ 1 ( B _ 1 )}{\textnormal{MMS} _ 1 ( n ,E)} =
\frac{2n}{n + 1}$, as required. 	

Next, for lower bound $\frac{2n\alpha}{ 2\alpha + 2n -3}$, let us consider an instance with $n$
agents and a set $E=\{ e_1,...,e_{2n^2-2n}\}$ of $2n^2-2n$ chores. We focus on agent 1 with cost
function $c _ 1( e _ j ) = 2 \alpha $ for $1\leq j \leq  n $ and $ c  _ 1 ( e _ j ) = 1 $ for $ j
\geq  n +1$. Consider the allocation $\mathbf{B} = (B _ 1, ...,B _ n )$ with $B _ 1 =\{ e _ 1,...,e
_ n \}, B _ 2 =\{e_{n+1},...,e_{3n-2} \}$ and $B_j=\{ e_{3n-1+(j-3)(2n-1)}, \ldots,
e_{3n-2+(j-2)(2n-1)}\}$ for any $j\geq 3$. Accordingly, bundle $B_2$ contains $2n-2$ chores and
$B_j$ contains $2n-1$ chores for any $j\geq 3$. For any agent $i\geq 2$, her cost functions is $ c_
i ( e ) = 0 $ for $ e \in B _ i $ and $c_ i ( e ) = 1 $ for $ e \in E\setminus  B _ i $.
Consequently, except for agent 1, no one else violate the condition of MMS and $\alpha$-EFX. As for
agent 1, his cost on $B_2$ is the smallest over all bundles and $c_ 1(B_1\setminus\{ e _ 1\}) =
2\alpha(n-1) = \alpha c _ 1 (B_ 2 )$, as a result, the allocation $\mathbf{B}$ is $\alpha$-EFX. For
$\textnormal{MMS}_1(n, E)$, it happens that $E$ can be evenly divided into $n$ bundles of the same
cost (for agent 1), so we have $\textnormal{MMS}_1(n, E) = 2\alpha + 2n-3$ implying the ratio
$\frac{c_1(B_1)}{\textnormal{MMS}_1(n, E)} = \frac{2n\alpha}{2\alpha + 2n-3}$, completing the
proof.
\end{proof}

The performance bound in Proposition~\ref{prop::4.5} is almost tight since  $\frac{n\alpha + n - 1}{ n -1
+ \alpha} - \frac{2n\alpha}{2\alpha + 2n - 3} < \frac{n-1}{n-1+\alpha} < 1$. In addition, we
highlight that the upper and lower bounds provided in Proposition~\ref{prop::4.5} are tight in two
interesting cases: (i) $\alpha=1$ and (ii) $n=2$.

On the approximation of EFX and EF1 for PMMS, we have the following propositions.

\begin{proposition}\label{prop::4.7}
When agents have additive cost functions, for any $\alpha \geq 1$, an $\alpha$-\textnormal{EFX}
allocation {is also $\frac{4\alpha}{ 2\alpha + 1}$-\textnormal{PMMS}}, and this guarantee is tight.
\end{proposition}

\begin{proof}
We first prove the upper bound. Let $\mathbf{A} = ( A _1, A _ 2, \ldots, A _
n )$ be an $\alpha$-EFX allocation and the approximation guarantee for PMMS is determined by agent $i$. We can assume $c_i(A_i)>0$; otherwise agent $i$ meets the condition of PMMS and we are done. Let $e ^ {*} $ be the chore in $A _ i $ having the minimum
cost for agent $i$, i.e., $ e ^ {*} \in \arg\min_{e \in A _ i } c _ i (e)$. Then, we divide the
proof into two cases.

\emph{Case 1}: $ |A _ i | = 1$. Then chore $ e ^ {*}$ is the unique element in $A _ i $, and thus $
c _ i ( e ^ {*}) = c _ i ( A _ i )$. By Lemma \ref{obs::3.1}, $ c _ i ( e ^
{*}) \leq \textnormal{MMS} _ i ( 2 , A _ i \cup A _ j)$ holds for any $j\in N\setminus \left\{ i
\right\}$. As a result, we have $ c _ i ( A_ i )  \leq \textnormal{MMS} _ i ( 2 , A _ i \cup A _
j), \forall j\in N\setminus \left\{ i \right\}$.

\emph{Case 2}: $ |A _ i | \geq 2$. Since $e ^ {*} \in \arg\min_{e \in A _ i } c _ i (e)$ and $
|A_i| \geq 2$, we have $ c _ i ( e ^ {*}) \leq \frac{1}{2} c _ i ( A_ i )$, and equivalently, $ c _
i ( A _ i \setminus \left\{ e ^ {*}\right\}) = c _ i (  A _i ) - c _ i ( e ^ {*}) \geq \frac{1}{2}
c _ i ( A_ i )$. Then, based on the definition of $\alpha$-EFX allocation, for any $j\in N\setminus
\left\{ i \right\}$, the following holds
\begin{equation}\label{eq::12}
	\alpha \cdot c _ i ( A_ j ) \geq c _ i ( A _i \setminus \left\{ e ^ {*}\right\}) \geq \frac{1}{2}
	\cdot c _ i ( A_ i).
\end{equation}
Combining Lemma \ref{obs::3.1} and Inequality (\ref{eq::12}), for any $j\in
N\setminus \left\{ i \right\}$, we have
\begin{equation*}
	\textnormal{MMS} _ i (2, A _ i \cup A _ j ) \geq \frac{1}{2} (c _ i ( A _ i ) + c _ i ( A_ j ))
	\geq \frac{2\alpha + 1}{4\alpha} c _ i ( A _ i ).
\end{equation*}
Therefore, for any $ j \in N\setminus \left\{ i \right\}$, $ c _ i ( A _ i ) \leq
\frac{4\alpha}{2\alpha + 1} \cdot \textnormal{MMS} _ i ( 2, A _ i \cup A _ j)$ holds, as required. 	

As for the tightness, consider an instance with $n$ agents and a set $E=\{ e_1,\ldots,e_{2n}\}$ of
$2n$ chores. Agents have identical cost profile and for every $i\in [n]$, $ c_i  (  e  _ 1 )  = c _
i ( e _ 2 ) = 2\alpha$ and $ c _ i ( e _ j )  =  1 $ for $3\leq j \leq  2n$. Consider the
allocation $\mathbf{B}=(B_1,\ldots,B_n)$ with $B_i=\{ e_{2i-1}, e_{2i}\}, \forall i \in N$. It is
not hard to verify that, except for agent 1, no one else would violate the condition of EFX and
PMMS. For agent 1, by removing any single chore from his bundle, the remaining cost is $\alpha$
times of the cost on others' bundle. Thus, allocation $\mathbf{B}$ is $\alpha$-EFX. Notice that for
any $j \geq 2$, bundle $B_1\cup B_j$ contains exactly two chores with cost $2\alpha$ and two chores
with cost 1, then $\textnormal{MMS}_1(2, B_1\cup B_j) = 2\alpha + 1$, implying for any $  j \neq
1$, $\frac{c_1(B_1)}{\textnormal{MMS}_1(2, B_1\cup B_j)} = \frac{4\alpha}{2\alpha+1}$.
\end{proof}

\begin{proposition}\label{prop::4.8}
When agents have additive cost functions, for any $\alpha \geq 1$, an $\alpha$-\textnormal{EF1}
allocation is also $\frac{2\alpha + 1}{\alpha+ 1}$-\textnormal{PMMS}, and this guarantee is tight.
\end{proposition}

\begin{proof}
We first prove the upper bound part. Let $\mathbf{A}= ( A _ 1,\ldots, A _ n)$ be an $\alpha$-EF1
allocation and the approximation guarantee for PMMS is determined by agent $i$. We can assume $c_i(A_i) >0 $; otherwise agent $i$ meets the condition of PMMS and we are done. To study PMMS, we fix another agent $j
\in N\setminus \left\{ i \right\}$, and let $e ^ {*} \in A _ i $ be the chore such that $ c _ i ( A
_ i \setminus \left\{ e ^ {*}\right\}) \leq \alpha \cdot c _ i  ( A _ j )$. We divide our proof
into two cases.

\emph{Case 1}: $ c _ i (e ^ {*}) > c _ i ( A _ i \cup A _ j \setminus \left\{ e ^ {*}\right\})$.
Consider $\left\{ \left\{ e ^ {*} \right\}, A _ i \cup A_j \setminus \left\{ e ^ {*}\right\}
\right\}$, a 2-partition of $A_i\cup A_j$. Since $ c _ i (e ^ {*}) > c _ i ( A _ i \cup A _ j
\setminus \left\{ e ^ {*}\right\})$, we can claim that this partition defining
$\textnormal{MMS}_i(2, A _ i \cup A _ j )$, and accordingly, $\textnormal{MMS}_i(2, A _ i \cup A _
j ) = c _ i ( e ^ {*})$ holds. From Lemma \ref{obs::3.1} and the definition of
$\alpha$-EF1, the following holds
\begin{equation}\label{eq::13}
	c _ i ( e ^ {*}) \geq \frac{1}{2} (c _ i ( A_ i ) +  c _ i ( A _j )) \geq \frac{1}{2} c _ i ( A _i )
	+ \frac{1}{2\alpha} \cdot c _ i ( A _ i \setminus\left\{ e ^ {*}\right\}).
\end{equation}
Then, based on (\ref{eq::13}) and the fact $\textnormal{MMS}_i(2, A _ i \cup A _ j ) = c _ i ( e ^
{*})$, we have
\begin{equation*}
\frac{c _ i ( A _ i )}{\textnormal{MMS}_i(2, A _ i \cup A _ j ) } \leq \frac{2\alpha + 1}{\alpha + 1}.
\end{equation*}

\emph{Case 2}: $ c _ i (e^ {*}) \leq c _ i ( A _ i \cup A _ j \setminus \left\{ e ^ {*}\right\})$.
By the definition of $\alpha$-EF1, we have $ c _ i ( A _ i \setminus \left\{ e ^ {*}\right\}) \leq
\alpha \cdot c _ i ( A_ j)$. As a consequence,
\begin{equation}\label{eq::14}
c _ i ( A_ i ) = c _ i ( e ^ {*}) + c _ i ( A _ i \setminus \left\{ e ^ {*}\right\})
\leq 2 c _ i ( A_ i \setminus \left\{ e ^ {*}\right\}) + c _ i ( A _j )
\leq (2\alpha + 1 ) \cdot c _ i ( A _j ),
\end{equation}
where the first inequality transition is due to $ c _ i (e ^ {*}) \leq c _ i ( A _ i \cup A _ j
\setminus \left\{ e ^ {*}\right\})$. Using Inequality (\ref{eq::14}) and additivity of cost
function, we have $c _ i ( A_ i ) \leq \frac{2\alpha + 1}{2\alpha + 2} \cdot c _ i ( A _ i\cup A _
j )$. By Lemma \ref{obs::3.1}, we have $\textnormal{MMS} _ i ( 2, A_ i \cup A_ j
) \geq \frac{1}{2} c _ i ( A_ i \cup A_ j )$ and then, the following holds,
\[
\frac{c_i ( A_ i )}{\textnormal{MMS} _ i ( 2, A_ i \cup A_ j ) } \leq \frac{2\alpha + 1}{\alpha + 1}.
\]

As for tightness, consider the following instance of $n$ agents and a set $E=\{
e_1,\ldots,e_{n+1}\}$ of $n+1$ chores. Agents have an identical cost profile and for every $i \in
[n]$, $ c_i ( e_1 ) = \alpha + 1, c _i ( e _ 2 ) = \alpha$ and $ c _i  ( e _ j ) = 1 $ for $ j \geq
3$. Then, consider the allocation $\mathbf{B} = (B_1,\ldots,B_n)$ with $B_1=\{e_1,e_2\}$ and $B_j =
\{ e_{j+1}\}, \forall j\geq 2$. It is not hard to verify that allocation $\mathbf{B}$ satisfying
$\alpha$-EF1, and moreover, the guarantee for PMMS is determined by agent 1. Notice that for any
$j\geq 2$, the combined bundle $B_1\cup B_j$ contains three chores with cost $\alpha+1, \alpha, 1$,
respectively. Thus, for any $j \geq 2$, we have $\textnormal{MMS}_1(2, B_1\cup B_j) = \alpha +1$,
implying the ratio $\frac{c_1(B_1)}{\textnormal{MMS}_1(2, B_1\cup B_j)} = \frac{2\alpha +1}{\alpha
+1}$.
\end{proof}

In addition to the approximation guarantee for PMMS, Proposition~\ref{prop::4.8} also has a direct
implication in approximating PMMS algorithmically. It is known that an EF1 allocation can be found
efficiently by allocating chores in a \emph{round-robin} fashion --- each of the agent $1,\ldots, n$ picks her most preferred item in that order, and repeat until all chores are assigned \citep{azizFairAllocationIndivisible2019}. Therefore,
Proposition~\ref{prop::4.8} with $\alpha =1$ leads to the following corollary, which is the only
algorithmic result for PMMS (in chores allocation), to the best of our knowledge.

\begin{corollary}
When agents have additive  cost functions, the round-robin algorithm outputs a
$\frac{3}{2}$-\textnormal{PMMS} allocation in polynomial time.
\end{corollary}

\section{Bounds on PMMS and MMS under additive setting}

Note that PMMS implies EFX in goods allocation according to
\citet{caragiannisUnreasonableFairnessMaximum2019}. This implication also holds in allocating
chores as stated in our proposition below.

\begin{proposition}\label{prop::5.1}
When agents have additive cost functions, a \textnormal{PMMS} allocation is also \textnormal{EFX}.
\end{proposition}

\begin{proof}
Let $\mathbf{A} = ( A _ 1, \ldots, A _ n )$ be a PMMS allocation. For the sake of contradiction, assume
$\mathbf{A}$ is not EFX and agent $i$ violates the condition of EFX, which implies $c_i(A_i) > 0$.

As agent $i$ violates the condition of EFX, there must exist an agent $ j \in N$ and $e ^ {*} \in A
_ i $ with $ c _ i ( e ^ {*}) > 0 $ such that $ c _ i ( A _ i \setminus \left\{ e ^ {*}\right\}) >
c _ i ( A _ j )$. Note chore $e ^ {*}$ is well-defined owing to $c_i(A_i) > 0$. Now, consider the
$2$-partition $\left\{ A _ i \setminus \left\{ e ^ {*} \right\}, A _ j \cup \left\{e ^ {*} \right\}
\right\} \in \Pi_2(A _ i \cup A _ j )$. By $c _ i ( A _ i \setminus \left\{e^{*}\right\}) > c _ i (
A _ j )$, the following holds:
\begin{equation}\label{eq::16}
	\begin{aligned}
		c _ i ( A _i ) &> \max \left\{ c _ i ( A _ i \setminus \left\{ e ^ {*}\right\}),
		c _ i ( A _ j \cup \left\{ e ^{*}\right\})  \right\} \\
		&\geq \min \limits_{\mathbf{B} \in \Pi _ 2 ( A_ i \cup A _ j )} \max
		\left\{ c _ i ( B _ 1), c _ i ( B _ 2)\right\} \geq c _ i ( A _i ),
	\end{aligned}
\end{equation}
where the last transition is by the definition of PMMS. Inequality (\ref{eq::16}) is a
contradiction, and therefore, $\mathbf{A}$ must be an EFX allocation.
\end{proof}

Since EFX implies EF1, Proposition~\ref{prop::5.1} directly leads to the following result.

\begin{proposition}\label{coro::5.1}
When agents have additive cost functions, a \textnormal{PMMS} allocation is also \textnormal{EF1}.
\end{proposition}

For approximate version of PMMS, when allocating goods it is shown in
\citet{amanatidisComparingApproximateRelaxations2018} that for any $\alpha$, $\alpha$-PMMS can imply
$\frac{\alpha}{2-\alpha}$-EF1. However, in the case of chores, our results indicate that
$\alpha$-PMMS has no bounded guarantee for EF1.

\begin{proposition}\label{prop::5.2}
When agents have additive cost functions, an $\alpha$-\textnormal{PMMS} allocation with
$1<\alpha\leq 2$ is not necessarily $\beta$-\textnormal{EF1} for any $\beta \geq 1$.
\end{proposition}

\begin{proof}
It suffices to show an $\alpha$-PMMS allocation with $\alpha \in (1,2)$ can not have a bounded
guarantee for the notion of EF1. Consider an instance with $n$ agents and $n+1$ chores $e_1$
\ldots, $e_{n+1}$. Agents have identical cost profile and for any $i$, we let $c_i(e_1) =
\frac{1}{\alpha -1}$, $c_i(e_2) = 1$ and $c_i(e_j) = \epsilon$ for $ 3\leq j \leq n+1$ where
$\epsilon$ takes arbitrarily small positive value. Then, consider an allocation $\mathbf{B} =
(B_1,\ldots,B_n)$ with $B_1= \{ e_1,e_2\}$ and $B_j=\{ e_{j+1}\}$ for $ 2 \leq j \leq n$.
Consequently, except for agent 1, other agents violate neither EF1 nor $\alpha$-PMMS. As for agent
1, notice that $\frac{1}{\alpha-1} > 1+\epsilon$ and thus, for any $j\geq 2$, the combined bundle
$B_1\cup B_j$ admits $\textnormal{MMS}_1(2, B_1\cup B_j) = \frac{1}{\alpha -1}$ implying
$\frac{c_1(B_1)}{\textnormal{MMS}_1(2, B_1\cup B_j) } = \alpha$. Thus, allocation $\mathbf{B}$ is
$\alpha$-PMMS. For the guarantee on EF1, as $c_1(B_j) = \epsilon$ for any $j\geq 2$, then removing
the chore with the largest cost from $B_2$ still yields the ratio
$\frac{c_1(B_1\setminus\{e_1\})}{c_1(B_j)} = \frac{1}{\epsilon} \rightarrow \infty$ as $\epsilon
\rightarrow 0$.
\end{proof}

Since for any $\alpha \geq 1$, $\alpha$-EFX is stricter than $\alpha$-EF1, the impossibility result
on EF1 in Proposition~\ref{prop::5.2} is also true for EFX.

\begin{proposition}\label{prop::5.3}
When agents have additive cost functions, an $\alpha$-\textnormal{PMMS} allocation
with $1 < \alpha \leq 2$ is not necessarily a $\beta$-\textnormal{EFX} allocation for any $\beta
\geq 1$.
\end{proposition}

We now study the approximation guarantee of PMMS for MMS. Since these two notions coincide when
{there are} only two agents, we consider the situation where $n\geq 3$. We first provide a tight
bound for $n=3$ and then give an almost tight bound for general $n$.

\begin{proposition}\label{prop:my-add1}
When agents have additive cost functions, for $n=3$, a \textnormal{PMMS} allocation is also
$\frac{4}{3}$-\textnormal{MMS}, and moreover, this bound is tight.
\end{proposition}

\begin{proof}
See Appendix \ref{proof of prop:my-add1}.
\end{proof}

For general $n$, we use the connections between PMMS, EFX and MMS to find the approximation
guarantee of PMMS for MMS. According to Proposition~\ref{prop::5.1}, a PMMS allocation is also EFX,
and by Proposition~\ref{prop::4.5}, EFX implies $\frac{2n}{n+1}$-MMS. As a result, we can claim
that PMMS also implies $\frac{2n}{n+1}$-MMS. With the following proposition we show that this
guarantee is almost tight.

\begin{proposition}\label{prop:my-add2}
When agents have additive cost functions, for $n\geq 4$, a \textnormal{PMMS} allocation is
$\frac{2n}{n+1}$-\textnormal{MMS} but not necessarily $\beta$-\textnormal{MMS} for any $\beta < \frac{2n+2}{n+3}$.
\end{proposition}

\begin{proof}
The positive part directly follows from Propositions~\ref{prop::5.1} and \ref{prop::4.5}. As for
the lower bound, consider an instance with $n$ (odd) agents and a set $E = \{ e _1,\ldots,e
_{2n}\}$ of $2n$ chores. We focus on agent 1 and his cost function is $c_1 ( e _ j ) =
\frac{n+1}{2}$ for $   1 \leq j \leq n$ and $ c _ 1 ( e_ j ) = 1 $ for $  n +  1 \leq  j \leq 2n$.
Consider the allocation $\mathbf{B} = ( B_1,\ldots,B_n)$ with $B_1 = \{ e_1, e_2\}$, $B_n= \{
e_{n+1},\ldots,e_{2n}\}$ and $ B_j =\{ e_{j+1} \}$ for any $j =2,\ldots,n-1$. For agents $i \geq
2$, her cost function is $ c_  i ( e )  = 0 $ for any $ e \in B _ i $ and $ c_ i (e) = 1 $ for any
$e \in E\setminus B  _ i $, and thus agent $ i $ has zero cost under allocation $\mathbf{B}$. As a
result, except for agent 1, other agents violate neither MMS nor PMMS. For agent 1, we have
$c_1(B_1) \leq \textnormal{MMS}_1(2, B_1\cup B_j)$ holds for any $ j \geq 2$, which implies
allocation $\mathbf{B}$ is PMMS. For $\textnormal{MMS}_1(n, E)$, it happens that $E$ can be evenly
divided into $n$ bundles of the same cost (for agent 1), so we have $\textnormal{MMS}_1(n,E) =
\frac{n+3}{2}$ yielding the ratio $\frac{c_1(B_1)}{\textnormal{MMS}_1(n,E) } = \frac{2n+2}{n+3}$.
\end{proof}

Next, we investigate the approximation guarantee of approximate PMMS for MMS. Let us start with an
example of six chores $E = \{ e_1,\ldots,e_6\}$ and three agents. We focus on agent 1 and the cost
function of agent 1 is $c_1(e_j)=1$ for $j=1,2,3$ and $c_1(e_j) = 0$ for $j = 4,5,6$, thus clearly,
$\textnormal{MMS}_1(3, E) = 1$. Consider an allocation $\mathbf{A}=(A_1,A_2,A_3)$ with
$A_1=\{e_1,e_2,e_3\}$. It is not hard to verify that allocation $\mathbf{A}$ is a
$\frac{3}{2}$-PMMS allocation and also a 3-MMS allocation. Combining the result in
Lemma~\ref{lemma::3.2}, we observe that allocation $\mathbf{A}$ only has a trivial guarantee on the
notion of MMS. {Motivated} by this example, we focus on $\alpha$-PMMS allocations with $\alpha <
\frac{3}{2}$.

\begin{proposition}\label{prop:4.9}
When agents have additive cost functions, for $n\geq 3$ and $1< \alpha < \frac{3}{2}$, an
$\alpha$-\textnormal{PMMS} allocation is $\frac{n\alpha}{ \alpha +
(n-1)(1-\frac{\alpha}{2})}$-\textnormal{MMS}, but not necessarily $(\frac{n\alpha}{\alpha +
(n-1)(2-\alpha)}-\epsilon)$-\textnormal{MMS} for any $\epsilon > 0$.
\end{proposition}

Before we can prove the above proposition, we need the following two lemmas.

\begin{lemma}\label{lemma::5.7}
For any $i\in N$ and $S\subseteq E$, suppose $\textnormal{MMS}_i(2, S)$ is defined by a
2-partition $\mathbf{T}=(T_1,T_2)$ with $c_i(T_1) = \textnormal{MMS}_i(2, S)$. If the number of
chores in $T_1$ is at least two, then $\frac{c_i(S)}{\textnormal{MMS}_i(2,S)} \geq \frac{3}{2}$.
\end{lemma}

\begin{proof}
For the sake of contradiction, we assume $\frac{c_i(S)}{\textnormal{MMS}_i(2,S)} < \frac{3}{2}$.
Since $c_i(T_1) = \textnormal{MMS}_i(2, S)$, we have $c_i(T_1) > \frac{2}{3}c_i(S)$, and
accordingly, $c_i(T_2) < \frac{1}{3} c_i(S)$ due to additivity. Thus, $c_i(T_1) - c_i(T_2) >
\frac{1}{3} c_i(S)$ holds, and we claim that each single chore in $T_1$ has cost strictly larger
than $\frac{1}{3}c_i(S)$ for agent $i$; otherwise, by moving the chore with the smallest cost in
$T_1$ to $T_2$, one can find a 2-partition in which the cost of larger bundle is smaller than
$c_i(T_1)$, contradiction. Based on our claim, we have $|T_1| = 2$. Notice that for any $e \in
T_1$, $c_i(e) >c_i(T_2)$ holds. As a result, moving one chore from $T_1$ to $T_2$ results in a
2-partition, in which the cost of larger bundle is strictly smaller than $c_i(T_1)$, contradicting
to the construction of allocation $\mathbf{T}$.
\end{proof}

\begin{lemma}\label{lemma::5.8}
For any $i\in N$ and $S_1, S_2 \subseteq E$, if $\textnormal{MMS}_i(2, S_1 \cup S_2) >
\textnormal{MMS}_i(2, S_1)$, then $\textnormal{MMS}_i(2, S_1 \cup S_2) \leq \frac{1}{2} c_i(S_1) +
c_i(S_2)$.
\end{lemma}

\begin{proof}
Suppose $\textnormal{MMS}_i(2, S_1)$ is defined by partition $(T_1,T_2)$ and we have
$\textnormal{MMS}_i(2, S_1) = c_i(T_1)$. We distinguish two cases according to the value of
$c_i(T_1)$. If $c_i(T_1) = \frac{1}{2} c_i(S_1)$, then consider $(T_1 \cup S_2 , T_2)$, a
2-partition of $S_1\cup S_2$. Clearly, $\textnormal{MMS}_i(2, S_1 \cup S_2) \leq c_i(T_1 \cup S_2)
= \frac{1}{2} c_i(S_1) + c_i(S_2)$. If $c_i(T_1) > \frac{1}{2} c_i(S_1)$, since
$\textnormal{MMS}_i(2, S_1 \cup S_2) > \textnormal{MMS}_i(2, S_1)$, we can claim that $c_i(T_1) -
c_i(T_2) < c_i(S_2)$; otherwise, considering partition $\{ T_1, T  _  2  \cup S _ 2 \}$, we have $
\textnormal{MMS} _ i ( 2 , S _ 1 \cup S _ 2 ) \leq  c _ i  ( T _ 1 ) = \textnormal{MMS} _ i (  2, S
_ 1 )$, a contradiction. Now let us consider $\{T_2\cup S_2, T_1\}$, another 2-partition of
$S_1\cup S_2$. According to our claim, we have $c_i(T_2 \cup S_2) > c_i(T_1)$, and thus,
$\textnormal{MMS}_i(2, S_1 \cup S_2) \leq c_i(T_2\cup S_2) < \frac{1}{2}c_i(S_1) + c_i(S_2)$, where
the last inequality is due to $c_i(T_2) = c_i(S_1) - c_i(T_1) < \frac{1}{2}c_i(S_1)$.
\end{proof}


\begin{proof}[Proof of Proposition~\ref{prop:4.9}]
We first prove the upper bound. Let $\mathbf{A} = (A_1,...,A_n)$ be an $\alpha$-PMMS allocation and
we focus our analysis on agent $i$. Let $\alpha ^{(i)} = \max_{j\neq i}
\frac{c_i(A_i)}{\textnormal{MMS}_i(2,A_i \cup A_j)}$ and $j^{(i)} $ be the index such that
$\textnormal{MMS}_i(2, A_i\cup A_{j^{(i)} }) \leq \textnormal{MMS}_i(2, A_i\cup A_{j})$ for any
$j\in N \setminus\{ i \}$ (tie breaks arbitrarily). By such a construction, clearly, $\alpha =
\max_{i \in N} \alpha ^{(i)}$ and $c_i(A_i) = \alpha ^{(i)} \cdot \textnormal{MMS}_i(2, A_i \cup
A_{j^{(i)}})$. Then, we split our proof into two different cases.

\emph{Case 1:} $\exists j \neq i$ such that $\textnormal{MMS}_i(2, A_i \cup A_j) =
\textnormal{MMS}_i(2, A_i)$. Then $\alpha ^{(i)} = \frac{c_i(A_i)}{\textnormal{MMS}_i(2, A_i)}$
holds. Suppose $\textnormal{MMS}_i(2, A_i)$ is defined by the 2-partition $(T_1,T_2)$ with
$c_i(T_1) = \textnormal{MMS}_i(2, A_i)$. If $|T_1| \geq 2$, by Lemma \ref{lemma::5.7}, we have
$\alpha ^{(i)} = \frac{c_i(A_i)}{\textnormal{MMS}_i(2, A_i)} \geq \frac{3}{2}$, contradicting to
$\alpha ^{(i)} \leq \alpha < \frac{3}{2}$. As a result, we can further assume $|T_1| = 1$. Then, by Lemma \ref{obs::3.1}, we have $\textnormal{MMS}_i(n, E) \geq c_i(T_1)$ and
accordingly, $\frac{c_i(A_i)}{\textnormal{MMS}_i(n, E)} \leq \frac{c_i(A_i)}{c_i(T_1)} = \alpha
^{(i)} \leq \alpha$. For $1<\alpha<\frac{3}{2}$ and $n \geq 3$, it is not hard to verify that
$\alpha \leq \frac{n\alpha}{ \alpha + (n-1)(1-\frac{\alpha}{2})}$, completing the proof for this
case. 	

\emph{Case 2:} $\forall j \neq i $, $\textnormal{MMS}_i(2, A_i \cup A_j) > \textnormal{MMS}_i(2,
A_i)$ holds. According to Lemma \ref{lemma::5.8}, for any $j\neq i$, the following holds
\begin{equation}\label{ieq::17}
	\textnormal{MMS}_i(2, A_i\cup A_j) \leq \frac{1}{2} c_i(A_i) + c_i(A_j).
\end{equation}
Due to the construction of $\alpha^{(i)}$, for any $j \neq i$, we have $c_i(A_i) \leq \alpha ^{(i)}
\cdot \textnormal{MMS}_i(2, A_i \cup A_j)$. Combining Inequality (\ref{ieq::17}), we have $c_i(A_j)
\geq \frac{2-\alpha ^{(i)}}{2\alpha ^{(i)}}c_i(A_i)$ for any $ j \neq i $. Thus, the following
holds,
\begin{equation}\label{ieq::18}
	\frac{c_i(A_i)}{\textnormal{MMS}_i(n, E)} \leq \frac{nc_i(A_i )}{ c_i (E)} \leq
	\frac{nc_i(A_i)}{c_i(A_i) + (n-1)\frac{2-\alpha^{(i)}}{2\alpha^{(i)}} c_i(A_i)}.
\end{equation}
The last expression in (\ref{ieq::18}) is monotonically increasing in $\alpha^{(i)}$, and
accordingly, we have
$$
\frac{c_i(A_i)}{\textnormal{MMS}_i(n, E)} \leq \frac{n\alpha}{\alpha +
	(n-1)(1-\frac{\alpha}{2})}.
$$ 	

As for the lower bound, consider an instance of $n$ (even) agents and a set $E=\{
e_1,...,e_{n^2}\}$ of $n^2$ chores. Agents have identical cost functions and for any $ i $, we let
$c_i (e_j) = \alpha$ for $1 \leq j \leq n $ and $c_ i (e_j) = 2 - \alpha$ for $n  +1 \leq j \leq n
^ 2$. Consider the allocation $\mathbf{B} = (B_1,...,B_n)$ with $B_i = \{ e_{ (i-1)n +1 },...,
e_{ni} \}$ for any $i \in [n]$. Since $\alpha > 1$, it is not hard to verify that, except for agent
1, no one else violates the condition of PMMS, and accordingly, the approximation guarantee for PMMS
is determined by agent 1. For agent 1, since $n$ is even, $\textnormal{MMS}_1(2,B_1\cup B_j) = n$
holds for any $j \geq 2$, and due to $c_1(B_1) = n\alpha$, we can claim that the allocation
$\mathbf{B}$ is $\alpha$-PMMS. Moreover, it is not hard to verify that $\textnormal{MMS}_1(n, E) =
\alpha + (n-1)(2-\alpha)$ and so $\frac{c_ 1(B _  1)}{\textnormal{MMS}_1(n, E)} =
\frac{n\alpha}{\alpha + (n-1)(2-\alpha)}$, completing the proof.
\end{proof}

The motivating example right before Proposition~\ref{prop:4.9}, unfortunately, only works for the
case of $n=3$. When $n$ becomes larger, an $\alpha$-PMMS allocation with $ \alpha \geq \frac{3}{2}$
is still possible to provide a non-trivial approximation guarantee on the notion of MMS.

We remain to consider the approximation guarantee of MMS for other fairness criteria. Notice that
all of EFX, EF1 and PMMS can have non-trivial guarantee for MMS (i.e., better than $n$-MMS).
However, the converse is not true and even the exact MMS does not provide any substantial guarantee
for the other three criteria.

\begin{proposition}\label{prop::5.6}
When agents have additive cost functions, for any $n\geq 3$, an \textnormal{MMS} allocation is not
necessarily $\beta$-\textnormal{PMMS} for any $1\leq \beta < 2$.
\end{proposition}

\begin{proof}
Consider an instance with $n$ agents and $p+2n-1$ chores denoted as $\{e_1,\ldots,e_{2n+p-1}\}$
where $p \in \mathbb{N}^+$ and $ p \gg 1$. We focus on agent 1 and his cost function is: $ c _  1 (
e_  j ) = 1$ for any $ 1 \leq j \leq n  + p$ and $ c  _ 1 ( e _ j ) = p $ for any $  j \geq  n + p
+ 1 $. Consider allocation $\mathbf{B} = ( B _1, \ldots,B _n)$ with $B _ 1 =\left\{ e _ 1, \ldots,
e _ {p+ 1}\right\}$, $B _ i =\left\{ e_{p+i}\right\}, \forall i =2,\ldots,n-2 $, $B_{n-1} = \{
e_{n+p-1}, e_{n+p}\}$ and $B_n = \{ e_{n+p+1},\ldots,e_{2n+p-1}\}$. For any agent $i\geq 2$, her
cost function is $ c _i (e) = 0  $ for any $ e \in B _ i $ and $  c  _ i ( e   ) =  1 $ for any $e
\notin B _ i $. Consequently, except for agent 1, other agents violate neither MMS nor PMMS, and
accordingly the approximation guarantee for PMMS and MMS is determined by agent 1. For
$\textnormal{MMS}_1(n, E)$, it happens that $E$ can be evenly divided into n bundles of the same
cost (for agent 1), so we have $\textnormal{MMS}_1(n, E) = p+1$. Accordingly, $c_1(B_1) =
\textnormal{MMS}_1(n, E) $ holds and thus, allocation $\mathbf{B}$ is MMS. As for the approximation
guarantee on PMMS, consider the combined bundle $B_1\cup B_2$ and it is not hard to verify that
$\textnormal{MMS}_1(2, B_1\cup B_2) = \lceil \frac{p+2}{2} \rceil$ implying
$\frac{c_1(B_1)}{\textnormal{MMS}_1(2, B_1\cup B_2)} = \frac{p+1}{\lceil \frac{p+2}{2} \rceil}
\rightarrow 2$ as $p \rightarrow \infty$.
\end{proof}

\begin{proposition}\label{prop::5.7}
When agents have additive cost functions, an \textnormal{MMS} allocation is not necessarily
$\beta$-\textnormal{EF1} or $\beta$-\textnormal{EFX} for any $\beta \geq1$.
\end{proposition}

\begin{proof}
By Proposition~\ref{prop::4.3}, the notion $\beta$-EFX is stricter than $\beta$-EF1, and thus, we
only need to show the unbounded guarantee on EF1. Again, we consider the instance given in the
proof of Proposition \ref{prop::5.6}. As stated in that proof, $\mathbf{B}$ is an MMS allocation,
and except for agent 1, no one else will violate the condition of PMMS. Note that PMMS is stricter
than EF1, then no one else will violate the condition of EF1. As for agent 1, each chore in $B_1$
has the same cost for him, so we can remove any single chore in $B_1$ and check its performance in
terms of EF1. When comparing to bundle $B_2$, we have $\frac{c_1(B_1\setminus \{ e_1\})}{c_1(B_2)}
= p \rightarrow \infty$ as $p\rightarrow \infty$.
\end{proof}

\section{Bounds beyond additive setting}

The results in previous sections demonstrate the strong connections between the four (additive)
relaxations of envy-freeness in the setting of additive cost functions. Under this umbrella, what
would also be interesting is that whether there still exists certain connections when agents' cost
profile is no longer additive. In this section, we also study the connections between fairness
criteria and, instead of additive cost functions, we assume that agents have submodular cost
functions, which have also been widely concerned with in fair division literature
\citep{ghodsiFairAllocationIndivisible2018, chanMaximinAwareAllocationsIndivisible2019}.

As a starting point, we consider EF, the strongest notion in the setting of additive, and see
whether it can still provide guarantee on other fairness notions. According to the definitions, the
notion of EF is, clearly, still stricter than EFX and EF1 if cost functions are monotone. Then, we
study the approximation guarantee of EF on MMS and PMMS. As shown by our results below, in contrast
to the results under additive setting, PMMS and MMS are no longer the relaxations of EF, and even
worse, the notion of EF does not provide any substantial guarantee on PMMS and MMS.

\begin{proposition}\label{prop::6.1}
When agents have submodular cost functions, an \textnormal{EF} allocation is not necessarily
$\beta_1$-\textnormal{MMS} or $\beta _ 2$-\textnormal{PMMS} for any $1\leq \beta  _ 1 < n$, $1\leq
\beta_2 < 2$.
\end{proposition}

\begin{proof}
It suffices to show that there exists an EF allocation with approximation guarantee $n$ and 2 for
MMS and PMMS, respectively. Consider an instance with $n$  (even) agents and a set $E$ of chores
with $|E| = n^ {2}$. Chores are placed in the form of $n\times n$  matrix
$E=\left[e_{ij}\right]_{n\times n}$. All agents have an identical cost function
$c(S)=\sum_{i=1}^{n} \min \left\{\left|E_{i} \cap S\right|, 1\right\}$ for any $S \subseteq E$,
where $E_i$ is the set of all elements in the $i$-th row of matrix $E$, i.e., $E_i=\{ e _ { i 1},
\ldots, e _ { i n} \}$. Since capped cardinality function $\left|E_{i} \cap S\right|$ of $S
\subseteq E$ is monotone and submodular for any fix $i$ ($1\le i\le n$), it follows that $c(\cdot)$
is also monotone and submodular.\footnote{More generally, if $f(\cdot)$ is submodular, then
$g(f(\cdot))$ is also submodular for any $g(\cdot)$ that is non-decreasing and concave.
Furthermore, conical combination (with sum as a special case) of submodular functions is also
submodular.}

Next, we prove that this instance permits an EF allocation, with which the approximation guarantee
for MMS and PMMS is $n$ and 2, respectively. Consider an allocation $\mathbf{B} = (B _ 1, ..., B_ n
)$ where for any $j$, bundle $B_j$ contains all elements in the $j$-th column of matrix E, i.e.,
$B_ j = (e _ {1j}, e _ { 2j},..., e _ { n j })$. One can compute that $c(B_j) = \sum_{i =1}^{n}
\min \left(\left|E_{i} \cap B_ j \right|, 1\right) = n$ holds for any $j \in [n]$, which implies
that allocation $\mathbf{B}$ is EF. Next, we check the approximation guarantee of $\mathbf{B}$ on
MMS. With a slightly abuse of notation, we let $\mathbf{E}$ be the allocation defined by
$n$-partition $E_1,...,E_n$, i.e., $\mathbf{E} = ( E_1,...,E_n)$. It is not hard to see that for
any $i\in N, c(E_i) = 1$. Then we claim that allocation $\mathbf{E}$ defines MMS for all agents;
otherwise, there exists another allocation in which each bundle has cost strictly smaller than 1,
and this never happens because $c  ( e ) = 1 $ for any $  e  \in E$ and $c(\cdot)$ is monotone.
Therefore, for any $i\in N, \textnormal{MMS}_i(n, E) = 1$, which implies
$\frac{c(B_i)}{\textnormal{MMS}_i(n, E)} = n$, as required.

Next, we argue that the allocation $\mathbf{B}$ is 2-PMMS. Fix $i,j \in N$ and $j \neq i$. Notice
that the combined bundle $B_i\cup B_j$ contains two columns of chores, so we can consider another
allocation $\mathbf{B}^{\prime} = (B ^{\prime}_ i,B ^{\prime}_ j)$ with $B ^{\prime}_  i = \left\{
e _ {1i}, \ldots, e _ {\frac{n}{2} i },e _ {1j},..., e _ {\frac{n}{2} j } \right\}$ and $B
^{\prime}_j = \{ e _ {\frac{n}{2} +1 i},..., e _ {n i },e _ {\frac{n}{2} +1 j},..., e _ {n j } \}$.
The idea of $\mathbf{B}^{\prime}$ is to split each column into two parts with equal size and one
part staring from the first row to $\frac{n}{2}$-th row while the other one containing the rest
half. By the definition of cost function $c(\cdot)$, we know $c(B^{\prime}_i)=c(B^{\prime}_j) =
\frac{n}{2}$ implying $\textnormal{MMS} _ i(2, B _ i \cup B _ j ) \leq \max \{ c (B^{\prime}_i), c
(B^{\prime}_j )\}  = \frac{1}{2} c _ i (B _ i)$. Therefore, $\mathbf{B}$ is a 2-PMMS allocation.
\end{proof}

In the aspect of worst-case analysis, combining Lemma~\ref{lemma::3.2} and
Proposition~\ref{prop::6.1}, EF can only have a trivial guarantee ($n$ and 2, respectively) on MMS
and PMMS, which is a sharp contrast to the results in additive setting where EF is strictly
stronger than these two notions. As we mentioned above, EF is stricter than EFX and EF1, then we
can directly argue that neither EFX nor EF1 can have better guarantees than trivial ones, namely,
2-PMMS and $n$-MMS.

\begin{proposition}\label{prop::sub_EFX_PMMS_MMS}
When agents have submodular cost functions, an \textnormal{EFX} allocation is not necessarily
$\beta _ 1$-\textnormal{MMS} or $\beta_2$-\textnormal{PMMS} for any $1 \leq \beta _  1 < n $, $1
\leq \beta_ 2 < 2$.
\end{proposition}

\begin{proposition}\label{prop::sub_EF1_PMMS_MMS}
When agents have submodular cost functions, an \textnormal{EF1} allocation is not necessarily
$\beta _ 1$-\textnormal{MMS} or $\beta_2$-\textnormal{PMMS} for any $ 1 \leq \beta_1 < n $, $1
\leq \beta _ 2 < 2$.
\end{proposition}

As for the connections between EFX and EF1, the statement of Proposition~\ref{prop::4.3} is still
true in the case of submodular.

\begin{proposition}\label{prop::6.2}
When agents have submodular cost functions, an $\alpha$-\textnormal{EFX} allocation is also
$\alpha$-\textnormal{EF1} for any $\alpha \geq 1$. On the other hand, an \textnormal{EF1}
allocation is not necessarily a $\beta$-\textnormal{EFX} for any $\beta \geq 1$.
\end{proposition}

\begin{proof}
The positive part follows directly from definitions of EFX and EF1. As for the impossibility result, the
instance in the proof of Proposition \ref{prop::4.3} is established in the case of additive. Since
an additive function is also submodular, we also have such an impossibility result here.
\end{proof}

Next, we study the notion of PMMS in terms of its approximation guarantee on EFX and EF1. Recall
the results of Propositions~\ref{prop::5.1} and \ref{coro::5.1}, a PMMS allocation is stricter than
EFX and EF1 in the additive setting. However, in the case of submodular, this relationship does not
hold any more, and even worse, PMMS provides non-trivial guarantee on neither EFX nor EF1.

\begin{proposition}\label{prop::6.3}
When agents have submodular cost functions, a \textnormal{PMMS} allocation is not necessarily a
$\beta$-\textnormal{EF1} or $\beta$-\textnormal{EFX} allocation for any $\beta \geq 1$.
\end{proposition}

\begin{proof}
By Proposition~\ref{prop::6.2}, for any $\beta \geq 1$, $\beta$-EFX is stronger than $\beta$-EF1,
and thus it suffices to show the approximation guarantee for EF1 is unbounded. In what follows, we
provide an instance that has a PMMS allocation with only trivial guarantee on EF1.

Consider an instance with two agents and a set $E=\left\{ e_1, e_2, e_3\right\}$ of chores. Agents
have identical cost function $c(S)=\min\{|S|,2\}$. Since $|S|$ is monotone and submodular, it
follows that $c(\cdot)$ is also monotone and submodular (see Footnote~\thefootnote).

Next, we prove this instance having a PMMS allocation whose guarantee for EF1 is unbounded. Since
in total, we have three chores, and thus in any 2-partition there always exists an agent receiving
at least two chores. Thus, we can claim that $\textnormal{MMS}_i(2, E) = 2$ for any $i \in [2]$.
Then, consider an allocation $\mathbf{B} = (B _1, B _ 2)$ with $B _ 1 = E$ and $B_ 2 =\emptyset$.
Allocation $\mathbf{B}$ is PMMS since, for any $i\in[2]$, $\max \left\{ c(B _ 1), c (B  _2)\right\}
= \textnormal{MMS} _i (2,E) = 2$ holds. However, bundle $B _ 2$ is empty and so $c  _  1 ( B  _ 2 )
=c( B  _ 2  )  = 0$. Then, no matter which chore is removed from bundle $B _ 1$, agent 1 still has
a positive cost, which implies an unbounded approximation guarantee for the notion of EF1.
\end{proof}

The approximation guarantee of an MMS allocation for EFX, EF1 and PMMS can be directly derived from
the results in the additive setting. According to Propositions~\ref{prop::5.6} and \ref{prop::5.7},
in additive setting MMS does not provide non-trivial guarantee on all other three notions. Since
additive functions belong to the class of submodular functions, we directly have the following two
results.

\begin{proposition}\label{prop::sub_MMS_PMMS}
When agents have submodular cost functions, an \textnormal{MMS} allocation is not necessarily
$\beta$-\textnormal{PMMS} for any $1 \leq \beta < 2$.
\end{proposition}

\begin{proposition}\label{prop::sub_MMS_EFX_EF1}
When agents have submodular cost functions, an \textnormal{MMS} allocation is not necessarily a
$\beta$-\textnormal{EF1} or $\beta$-\textnormal{EFX} allocation for any $\beta \geq 1$.
\end{proposition}

At this stage, what remains is the approximation guarantee of PMMS on MMS. Before presenting the
main result, we provide a lemma, which states that the quantity of MMS is monotonically
non-decreasing on the set of chores to be assigned.

\begin{lemma}\label{Lemma::monotoneMMS}
Given a monotone function $c(\cdot)$ defined on ground set $E$, for any subsets $S \subseteq T
\subseteq E $, if quantities $\textnormal{MMS}(2, S)$ and $\textnormal{MMS}(2, T)$ are computed
based on function $c(\cdot)$, then $\textnormal{MMS} (2, S) \leq \textnormal{MMS}(2, T)$.
\end{lemma}

\begin{proof}
Let $\{ T _ 1, T _ 2 \}$ be the 2-partition of set $T$ and moreover it defines $\textnormal{MMS}(2,
T) = c ( T _ 1 ) \geq c ( T _ 2 )$. We then consider $ \{T _ 1 \cap S, T _ 2 \cap S \}$, which is,
clearly, a 2-partition of $S$ due to $S \subseteq T $. According to the definition of MMS, we have
$$
\textnormal{MMS}(2, S) \leq \max\{ c(T_1\cap S), c(T _ 2 \cap S )\}\leq \max\{ c(T_1), c(T_2)\}
= \textnormal{MMS}(2, T),
$$
where the second inequality transition is because $c(\cdot)$ is monotone.
\end{proof}

\begin{proposition}\label{prop::sub_PMMS_MMS}
When agents have submodular cost functions, for any $1\leq  \alpha \leq 2$, an
$\alpha$-\textnormal{PMMS} allocation is also $\min\{n,\alpha\lceil \frac{n}{2}
\rceil\}$-\textnormal{MMS}, and this guarantee is tight.
\end{proposition}

\begin{proof}
We first prove the upper bound. According to Lemma~\ref{lemma::3.2}, any allocation is $n$-MMS and
so what remains is to prove the upper bound of $\alpha \lceil \frac{n}{2} \rceil$. Fix agent $i$
with cost function $c_i(\cdot)$. Suppose $n$-partition $\{ T _ 1, \ldots, T _ n  \}$ defines
$\textnormal{MMS}_i(n, E)$ and w.l.o.g, we assume $ c _ i ( T _ 1 ) \geq c _ i ( T _ 2 )\geq \cdots
\geq c _ i ( T _ n )$, i.e., $ c _ i ( T _ 1 ) = \textnormal{MMS}_i(n, E)$. Then, we let
$2$-partition $\{ Q_1, Q _ 2\}$ defines $\textnormal{MMS}_i(2, E)$ and $ c _ i (Q_1) \geq c _ i (Q
_ 2 )$, i.e., $ c _ i (Q _ 1 ) = \textnormal{MMS}_i(2, E)$. Let $\mathbf{A}$ be an arbitrary
$\alpha$-PMMS allocation, and accordingly, for any $ j \neq i $, we have $ c _ i ( A _ i ) \leq
\alpha\cdot\textnormal{MMS} _ i ( 2, A _ i \cup A _ j )$. Since $A _ i \cup A _ j $ is a subset of
$E$, according to Lemma~\ref{Lemma::monotoneMMS}, we have $\textnormal{MMS}_i (2, A _ i \cup A _ j
) \leq \textnormal{MMS}_i ( 2, E ) $. We then construct an upper bound of $\textnormal{MMS}_i(2,
E)$ through partition $ \{ T _ 1, \ldots, T _ n\}$.

Let us consider a 2-partition $\{ B_1, B _ 2 \}$ of $E$ with $ B _ 1 = \{ T _ 1, T _ 2, \ldots, T _
{ \lceil \frac{n}{2} \rceil}\}$, $ B _ 2 = \{ T _{\lceil \frac{n}{2} \rceil  + 1}, \ldots, T _ n
\}$. Then, the following holds:
$$
\begin{aligned}
\max\{ c _ i ( B _ 1 ), c _ i ( B _ 2 )\}
    &= \max \{ c _ i ( \cup_{ j = 1} ^ { \lceil \frac{n}{2} \rceil} T _ j ), c _ i ( \cup _ { j
     = \lceil \frac{n}{2} \rceil + 1} ^ {n} T _ j  ) \} \\
	& \leq \max \{ \sum _ {j = 1} ^ {\lceil \frac{n}{2} \rceil} c _ i ( T _ j ), \sum _ { j
     = \lceil \frac{n}{2} \rceil + 1} ^ {n} c _ i ( T _ j )\}
	 \leq \lceil \frac{n}{2} \rceil\cdot c _ i ( T _ 1 ),
\end{aligned}
$$
where the first inequality transition is due to subadditivity of $c_i(\cdot)$ and the second
inequality transition is because $ c _ i ( T _ 1 ) \geq c _ i ( T _ 2 ) \cdots \geq c _ i ( T _ n
)$. Recall $c _ i ( Q _ 1 ) = \textnormal{MMS}_i(2, E)\leq \max\{ c _ i ( B _ 1), c _ i ( B _ 2 )
\}$, and accordingly we have $ \textnormal{MMS}_i(2, E) \leq \lceil \frac{n}{2} \rceil \cdot c _ i
(T _ 1 ) = \lceil \frac{n}{2} \rceil \cdot \textnormal{MMS}_i( n ,E )$. Therefore, for any $j\neq
i$, the following holds:
$$
\frac{c _ i ( A _ i  )}{\textnormal{MMS} _ i ( n , E )} \leq
\frac{ \alpha\cdot\textnormal{MMS}_ i ( 2, A _ i \cup A _ j )}{\textnormal{MMS} _ i ( n, E )}
\leq \frac{ \alpha \cdot \textnormal{MMS}_i(2, E )}{\textnormal{MMS} _ i (n ,E )}
\leq \alpha \cdot \lceil \frac{n}{2} \rceil.
$$

As for the lower bound, it suffices to show that for any $\alpha \in [1,2]$, there exists an
$\alpha$-PMMS allocation with approximation guarantee $\alpha\lceil \frac{n}{2} \rceil $ of MMS
when $\alpha \lceil \frac{n}{2} \rceil \leq n $. Let us consider an instance with $n$ (even) agents
and a set $E$ of chores with $|E| = n(n+1)$. Since $\alpha \leq 2$ and $n$ is even, clearly we have
$ \alpha \lceil \frac{n}{2} \rceil \leq n$. Chores are placed in $n\times (n+1)$ matrix $E=\left[
e_{ij}\right]_{n\times(n+1)}$. For $j\in [n+1]$, denote by $ P_j $ the $j$-th column, i.e., $  P _
j = \{ e _ { 1j }, e _{2j}, \ldots, e _ {nj}\}$. We concentrate on allocation $\mathbf{A}$  with $A
_ 1 = P _ 1 \cup \cdots \cup P _{\lfloor \alpha\frac{n}{2} \rfloor} \cup P _ n $, $A _ j = \left\{
e_{ j, \lfloor \alpha\frac{n}{2} \rfloor + 1} , \ldots, e _{j, n-1}, e _ { j, n + 1}\right\}$ for any $2\leq j
\leq n -1$, and $  A _ n = \left\{ e _ {  n, \lfloor \alpha\frac{n}{2} \rfloor + 1}, \ldots,e _{n, n-1}, e _ {
n, n + 1} \right\} \cup \left\{ e _ {  1, \lfloor \alpha\frac{n}{2} \rfloor + 1}, \ldots, e _{1, n-1}, e _ { 1,
n + 1}\right\}$. For any $2 \leq i \leq n$, agent $i$ has additive cost function $ c _ i (\cdot)$
with $c _ i ( e ) = 0 $ for any $ e \in A _ i $, and $ c _ i ( e ) = 1  $ for any $ e \in E
\setminus A _ i $. Then, for every $2 \leq i \leq n$, agent $i$ has an additive, clearly monotone
and submodular, cost function, and violates neither PMMS nor MMS due to $ c _ i ( A _ i ) =0 $.
Consequently, the approximation guarantee of $\mathbf{A}$ on both PMMS and MMS are determined by
agent 1.

As for the cost function $c _ 1 ( \cdot ) $ of agent 1, for any $S \subseteq E$, we let
$$
c _ 1 ( S ) = \sum _ { j = 1} ^ {n-1} \min \{ | S \cap P _ j |, 1 \} +
\delta \cdot \min \{ | S \cap P _ n |, 1 \} + (1-\delta) \cdot \min\{ | S \cap P _{ n + 1}|, 1 \},
$$
where $\delta = \alpha\frac{n}{2} -\lfloor \alpha\frac{n}{2} \rfloor$. Function $c _ 1 (\cdot)$ is
clearly monotone. As in the proof of Proposition~\ref{prop::6.1} (see Footnote~\thefootnote), $c _
1 (\cdot)$ as a conical combination of submodular functions is also submodular.

We argue $\mathbf{A}$ is an $\alpha$-PMMS allocation with approximation guarantee
$\alpha\frac{n}{2}$ on the notion of MMS. In fact, under allocation $\mathbf{A}$, one can compute $
c_1 ( A _ 1 ) = \lfloor \alpha \frac{n}{2} \rfloor + \delta = \alpha \frac{n}{2} $ and  $c _ 1 ( A
_ 1 \cup A _ j ) = c _ 1(E) = n  $ for any $2 \leq j \leq n$. Then, for any $ j \geq 2$, due to
Lemma~\ref{obs::3.1}, it holds that $\textnormal{MMS}_1(2, A _ 1 \cup A _ j ) \geq \frac{n}{2}$,
which then imply $  c _ 1 ( A _ 1 ) \leq \alpha \textnormal{MMS}_1(2, A _ 1 \cup A _ j )$. Thus,
allocation $\mathbf{A}$ is $\alpha$-PMMS. As for the quantity of $\textnormal{MMS}_1(n, E)$,
consider partition $\{ B  _ i \}_{ i \in [n]}$ with $ B _ i = P _ i $ for $  1 \leq i \leq n  -1$
and $ B _ n = P _ n \cup P _ { n + 1 }$. It is not hard to verify $ c _ 1 ( B _ i ) = 1 $ for any $
i \in [n]$. According to Lemma~\ref{obs::3.1}, we have $\textnormal{MMS}_1(n ,E) \geq \frac{1}{n} c
_ 1 ( E ) =  1$. Hence, partition $\{B  _ i \}_{i\in [n]}$ defines $\textnormal{MMS}_1(n, E) = 1$,
and accordingly, the approximation guarantee of $\mathbf{A}$ for MMS is $\alpha \frac{n}{2}$,
equivalent to $ \alpha \lceil \frac{n}{2} \rceil$ since $n$ is even.
\end{proof}

We remark that all statements in this section are still true if agents have subadditive cost
functions. Results in this section show that although PMMS (or MMS) is proposed as relaxation of EF
under additive setting, there are few connections between PMMS (or MMS) and EF in the submodular
setting. This motivates new submodular fairness notions which is not only a relaxation of EF but
also inherit the spirit of PMMS (or MMS).

\section{Price of fairness under additive setting}

After having compared the fairness criteria between themselves, in this section we study the
efficiency of these fairness criteria in terms of the price of fairness with respect to social
optimality of an allocation.

\subsection{Two agents}

We start with the case of two players. Our first result concerns EF1.

\begin{proposition}\label{thm::7.1}
When $n=2$ and agents have additive cost functions, the price of \textnormal{EF1} is $5/4$.
\end{proposition}

\begin{proof}
For the upper bound part, we analyze the allocation returned by algorithm $ALG_1$, whose detailed
description is in Appendix~\ref{sec:EF1-algorithm}. In this proof, we denote $L(k) = \{ e_1,
\ldots, e_ k\}$ and $ R(k) = \{ e_k, \ldots , e _ m \}$. We first show that $ALG_1$ is well-defined
and can always output an EF1 allocation. Note that $\mathbf{O}$ is the optimal allocation for the
underlying instance due to the order of chores. We consider the possible value of index $s$.
Because of the normalized cost function, trivially, $s<m$ holds. If $s=0$, $ALG_1$ outputs the
allocation returned by round-robin (line 6) and clearly, it's EF1. If the optimal allocation
$\mathbf{O}$ is EF1 (line 9), we are done. For this case, we claim that if $s = m-1$, then
$\mathbf{O}$ must be EF1. The reason is that for agent 1, his cost $c_1(O_1) \leq c_2(O_2) \leq
c_1(O_2)$ where the first transition due to line 1 of $ALG_1$, and thus he does not envy agent 2.
For agent 2, since he only receives a single chore in optimal allocation due to $s=m-1$, clearly,
he does not violate the condition of EF1, either. Thus, allocation $\mathbf{O}$ is EF1 in the case
of $s= m -1$. Next, we study the remaining case (lines 11--13) that can only happen when $1\leq
s\leq m-2$. We first show that the index $f$ is well-defined. It suffices to show $c_2(R(s+2)) >
c_2(L(s))$. For the sake of contradiction, assume $c_2(R(s+2)) \leq c_2(L(s))$. This is equivalent
to $c_2(O_2\setminus \{ e_{s+1}\}) \leq c_2(O_1)$, which means agent 2 satisfying EF1 in allocation
$\mathbf{O}$. Due to the assumption (line 1), $c_1(O_1) \leq c_2(O_2) \leq c_1(O_2)$ holds, and
thus, agent 1 is EF under the allocation $\mathbf{O}$. Consequently, the allocation $\mathbf{O}$ is
EF1, contradiction. Then, we prove allocation $\mathbf{A}$ (line 13) is EF1. According to the order
of chores, it holds that
$$
\frac{c_1(L(f))}{c_2(L(f))} \leq \frac{c_1(R(f+2))}{c_2(R(f+2))}.
$$
Since $c_2(R(f+2)) > c_2(L(f)) \geq 0$, this implies,
$$
\frac{c_1(L(f))}{c_1(R(f+2))} \leq \frac{c_2(L(f))}{c_2(R(f+2))}.
$$
By the definition of index $f$, we have $c_2(R(f+2)) > c_2(L(f))$ and therefore $c_1(L(f)) <
c_1(R(f+2))$ which is equivalent to $c_1(A_1 \setminus \{ e_{f+1}\}) < c_1(A_2)$. Thus, agent 1 is
EF1 under allocation $\mathbf{A}$. As for agent 2, if $f = m-2$, then $A_2 = 1$ and clearly, agent
2 does not violate the condition of EF1. We can further assume $f\leq m -3$. Since $f$ is the
maximum index satisfying $f \geq s$ and $c_2(R(f+2)) > c_2(L(f))$, it must hold that $c_2(R(f+3))
\leq c_2(L(f+1))$, which is equivalent to $c_2(A_2\setminus \{ e_{f+2}\}) \leq c_2(A_1)$ and so
agent 2 is also EF1 under allocation $\mathbf{A}$. 	

Next, we show the social cost of the allocation returned by $ALG_1$ is at most 1.25 times of the
optimal social cost. If $s=0$, both agents have the same cost profile, then any allocations have
the optimal social cost and we are done in this case. If allocation $\mathbf{O}$ is EF1, then
clearly, we are done. The remaining case is of lines 11--13 of $ALG_1$. Since $c_1(O_1) \leq
c_2(O_2) \leq c_1(O_2)$, we have $c_1(O_1) \leq \frac{1}{2}$. Notice that $\mathbf{O}$ is not EF1,
then $c_2(O_2) > \frac{1}{2}$ must hold; otherwise, $c_2(O_2) \leq c_2(O_1)$ and allocation
$\mathbf{O}$ is EF, contradiction. Therefore, under the case where allocation $\mathbf{O}$ is not
EF1, we must have $c_1(O_1) \leq \frac{1}{2}$ and $c_2(O_2) > \frac{1}{2}$. Due to $f+2 \geq s+1$
and the order of chores, it holds that
$$
\frac{c_1(R(f+2))}{c_2(R(f+2))} \geq \frac{c_1(O_2)}{c_2(O_2)}.
$$
This implies $c_1(R(f+2)) \geq \frac{c_1(O_2)}{c_2(O_2)}c_2(R(f+2))$, and equivalently,
$$c_1(A_1) = c_1(L(f+1)) \leq 1-\frac{c_1(O_2)}{c_2(O_2)}c_2(R(f+2)).
$$
Again, by the construction of $f$, we have
$$
c_2(A_2) = c_2(R(f+2)) > c_2(L(f)) \geq c_2(L(s)) = c_2(O_1).
$$
Therefore, we derive the following upper bound,
\begin{equation}\label{upper bound of A}
	\begin{aligned}
		c_1(A_1) + c_2(A_2) &\leq 1 - (\frac{c_1(O_2)}{c_2(O_2)} - 1) c_2(A_2)
		\leq 	1 - (\frac{c_1(O_2)}{c_2(O_2)} - 1) c_2(O_1)\\
		&= 1 -(\frac{1 - c_1(O_1)}{c_2(O_2)} - 1) (1- c_2(O_2)),\\
	\end{aligned}
\end{equation}
where the second inequality is due to $\frac{c_1(O_2)}{c_2(O_2)} \geq 1$ and $c_2(A_2) \geq
c_2(O_1)$. Based on (\ref{upper bound of A}), we have an upper bound on the price of EF1 as
follows:
\begin{equation}\label{upper bound of price EF1}
\textnormal{Price of EF1} \leq \frac{1 -(\frac{1- c_1(O_1)}{c_2(O_2)} - 1) (1- c_2(O_2))}{c_1(O_1) + c_2(O_2)}.
\end{equation}
Recall $ 0 \leq c_1(O_1) \leq \frac{1}{2} < c_2(O_2) \leq 1$. The partial derivatives of the
faction in (\ref{upper bound of price EF1}) with respect to $c_1(O_1)$ is equal to the following:
$$
\frac{1}{(c_1(O_1) + c_2(O_2))^2} (\frac{1}{c_2(O_2)} - 2).
$$
It is not hard to see this derivative has a negative value for any $\frac{1}{2}<c_2(O_2) \leq1$.
Thus, the fraction in (\ref{upper bound of price EF1}) takes maximum value only when $c_1(O_1) = 0$
and hence,
$$
\textnormal{Price of EF1} \leq \frac{3 - \frac{1}{c_2(O_2)}}{c_2(O_2)} - 1.
$$
Similarly, by taking the derivative with respect to $c_2(O_2)$, the maximum value of this
expression happens only when $c_2(O_2) = \frac{2}{3}$, then one can easily compute the maximum
value of the RHS of (\ref{upper bound of price EF1}) is $1.25$. Therefore, the $\textnormal{price
of EF1} \leq 1.25$.

As for the lower bound, consider an instance with a set $E =\{ e_1,e_2,e_3\}$ of three chores. The
cost function of agent 1 is $c_1(e_1) = 0$ and $c_1(e_2) = c_1(e_3) = \frac{1}{2}$. For agent 2,
his cost is $c_2(e_1) = \frac{1}{3} - 2\epsilon$ and $c_2(e_2)= c_2(e_3) = \frac{1}{3} + \epsilon$
where $\epsilon>0$ takes arbitrarily small value. An optimal allocation assigns chore $e_1$ to
agent 1 and the rest chores to agent 2, which yields the optimal social cost $\frac{2}{3} +
\epsilon$. However, this allocation is not EF1 since agent 2 envies agent 1 even removing one chore
from his bundle. To achieve EF1, agent 2 can not receive both of chores $e_2$ and $e_3$, and so,
agent 1 must receive one of chore $e_2$ and $e_3$. Therefore, the best EF1 allocation can be
assigning chore $e_1$ and $e_2$ to agent 1 and chore $e_3$ to agent 2 resulting in the social cost
$\frac{5}{6} + \epsilon$. Thus, the price of EF1 is at least $\frac{\frac{5}{6} +
\epsilon}{\frac{2}{3} + 2\epsilon} \rightarrow \frac{5}{4}$ as $\epsilon \rightarrow 0 $,
completing the proof.
\end{proof}

According to Propositions~\ref{prop::4.6} and \ref{prop::4.8}, EF1 implies 2-MMS and
$\frac{3}{2}$-PMMS. The following two propositions confirm an intuition --- if one relaxes the
fairness condition, then less efficiency will be sacrificed.
\begin{proposition}\label{PoFn=2::MMS}
When $n=2$ and agents have additive cost functions, the price of 2-\textnormal{MMS} is 1.
\end{proposition}

\begin{proof}
The proof directly follows from Lemma~\ref{lemma::3.2}.
\end{proof}

\begin{proposition}\label{PoFn=2::PMMS}
When $n=2$ and agents have additive cost functions, the price of $\frac{3}{2}$-\textnormal{PMMS} is
$7/6$.
\end{proposition}

\begin{proof}
We first prove the upper bound. Given an instance $I$, let $\mathbf{O}=( O_1, O_2)$ be an optimal
allocation of $I$. If the allocation $\mathbf{O}$ is already $\frac{3}{2}$-PMMS, we are done. For
the sake of contradiction, we assume that agent 1 violates the condition of $\frac{3}{2}$-PMMS in
allocation $\mathbf{O}$, i.e., $c_1(O_1) > \frac{3}{2}\textnormal{MMS}_1(2,E)$. Suppose
$O_1=\{e_1,\ldots,e_h\}$ and the index satisfies the following rule; $\frac{c_1(e_1)}{c_2(e_1)}
\geq \frac{c_1(e_2)}{c_2(e_2)} \geq \cdots \geq \frac{c_1(e_h)}{c_2(e_h)}$. In this proof, for
simplicity, we write $L(k):= \{ e_1,...,e_k\}$ for any $1\leq k \leq h $ and $L(0) = \emptyset$.
Then, let $s$ be the index such that $c_1(O_1 \setminus L(s)) \leq \frac{3}{2}\textnormal{MMS}_1(2,
E)$ and $c_1(O_1 \setminus L(s-1)) > \frac{3}{2} \textnormal{MMS}_1(2, E)$. In the following, we
divide our proof into two cases. 	

\emph{Case 1:} $c_1(L(s)) \leq \frac{1}{2} c_1(O_1)$. Consider allocation $\mathbf{A} = (A_1, A_2)$
with $A_1 = O_1\setminus L(s)$ and $A_2 = O_2 \cup L(s)$. We first show allocation $\mathbf{A}$ is
$\frac{3}{2}$-PMMS. For agent 1, due to the construction of index $s$, he does not violate the
condition of $\frac{3}{2}$-PMMS. {As for agent 2, we claim that $c_2(A_2) = 1 - c_2(O_1 \setminus
L(s-1)) + c_2(e_s) < \frac{1}{4} + c_2(e_s)$ because $c_2(O_1\setminus L(s-1)) \geq
c_1(O_1\setminus L(s-1)) >\frac{3}{2}\textnormal{MMS}_1(2, E) \geq \frac{3}{4}$ where the first
inequality transition is due to the fact that $O_1$ is the bundle assigned to agent 1 in the
optimal allocation.} If $c_2(e_s) < \frac{1}{2}$, then clearly, $c_2(A_2) < \frac{3}{4}
\leq\frac{3}{2}\textnormal{MMS}_2(2, E)$. If $c_2(e_s) \geq \frac{1}{2}$, then $c_2(e_s) =
\textnormal{MMS}_1(2, E)$ and accordingly, it is not hard to verify that $c_2(A_2) \leq
\frac{3}{2}\textnormal{MMS}_1(2, E)$. Thus, $\mathbf{A}$ is a $\frac{3}{2}$-PMMS allocation. 	

Next, based on allocation $\mathbf{A}$, we derive an upper bound on the price of
$\frac{3}{2}$-PMMS. First, by the order of index, $\frac{c_1(L(s))}{c_2(L(s))} \geq \frac{
	c_1(O_1)}{c_2(O_1)}$ holds, implying $c_2(L(s)) \leq \frac{c_2(O_1)}{c_1(O_1)} c_1(L(s))$. Since
$A_1=O _ 1\setminus L(s)$ and $A_2 = O_2 \cup L(s)$, we have the following:
\[
	\begin{aligned}
		\textnormal{Price of } \frac{3}{2}\textnormal{-PMMS}
        &\leq 1 + \frac{c_2(L(s)) - c_1(L(s))}{c_1(O_1) + c_2(O_2)}
		\leq 1 + \frac{c_1(L(s)) (\frac{c_2(O_1)}{c_1(O_1)} - 1)}{c_1(O_1) + c_2(O_2)} \\
		& = 1 + \frac{ \frac{c_1(L(s))}{c_1(O_1)} (1- c_2(O_2) -c_1(O_1)) }{ c_1(O_1)+ c_2(O_2)} \\
		&\leq 1 + \frac{\frac{1}{2} - \frac{1}{2}(c_1(O_1) + c_2(O_2))}{c_1(O_1) + c_2(O_2)}
		\leq 1-\frac{1}{2}+\frac{1}{2}\times \frac{4}{3} = \frac{7}{6},
	\end{aligned}
\]
where the second inequality due to $c_2(L(s)) \leq \frac{c_2(O_1)}{c_1(O_1)} c_1(L(s))$; the third
inequality due to the condition of Case 1; and the last inequality is because $c_1(O_1) >
\frac{3}{2}\textnormal{MMS}_1(2, E) \geq \frac{3}{4}$. 	

\emph{Case 2:} $c_1(L(s)) > \frac{1}{2} c_1(O_1)$. We first {derive a} lower bound on $c_1(e_s)$.
{Since $ c _ 1 ( e _ s) = c _ 1 (O _ 1 \setminus L (s-1)) + c _ 1 ( L _ s ) - c_ 1( O_ 1)$, combine
which with the condition of Case 2 implying $ c _ 1 ( e _ s ) > c _ 1 ( O _ 1 \setminus L ( s - 1
)) - \frac{1}{2} c _ 1 ( O _ 1)$, and consequently we have $ c _ 1 ( e _ s ) > \frac{3}{2}
\textnormal{MMS} _ 1 (  2, E) - \frac{1}{2} c _ 1 ( O _ 1 ) \geq \frac{1}{4}$} where the last
transition {is} due to $\textnormal{MMS}_1(2, E) \geq \frac{1}{2}$ and $c_1(O_1) \leq 1$. Then, we
consider two subcases. 	

{If $0 \leq c_2(e_s) - c_1(e_s) \leq\frac{1}{8}$, consider} an allocation $\mathbf{A}
= (A_1,A_2)$ with $A_1 = O_1\setminus \{e_s\}$ and $A_2=O_2\cup \{ e_s\}$. We first show the
allocation $\mathbf{A}$ is $\frac{3}{2}$-PMMS. For agent 1, since $c_1(e_s) > \frac{1}{4}$,
$c_1(A_1) = c_1(O_1) -c_1(e_s) <\frac{3}{4} \leq \frac{3}{2}\textnormal{MMS}_1(2, E)$. As for agent
2, $c_2(A_2) = c_2(O_2) + c_2(e_s) \leq 1-c_1(O_1) + c_2(e_s) < \frac{1}{4} + c_2(e_s)$. If
$c_2(e_s) < \frac{1}{2}$, then clearly, $c_2(A_2)\leq\frac{3}{4}<\frac{3}{2}\textnormal{MMS}_2(2,
E)$ holds. If $c_2(e_s)\geq \frac{1}{2}$, we have $c_2(e_s) =\textnormal{MMS}_2(2, E)$ and
accordingly, it is not hard to verify that $c_2(A_2)\leq\frac{3}{2}\textnormal{MMS}_2(2, E)$. Thus,
the allocation $\mathbf{A}$ is $\frac{3}{2}$-PMMS. Next, based on the allocation $\mathbf{A}$, we
derive an upper bound regarding the price of $\frac{3}{2}$-PMMS,
\[
\textnormal{Price of } \frac{3}{2}\textnormal{-PMMS} \leq \frac{c_1(O_1) - c_1(e_s)
	+ c_2(O_2) +c_2(e_s)}{c_1(O_1)+c_2(O_2)}
\leq 1 + \frac{1}{8}\times \frac{4}{3} = \frac{7}{6},
\]
where the second inequality due to $0 \leq c_2(e_s) - c_1(e_s) \leq\frac{1}{8}$ and $c_1(O_1) >
\frac{3}{4}$.

If $c_2(e_s) - c_1(e_s) > \frac{1}{8}$, consider an allocation {$\mathbf{A ^ {\prime}} = (A ^
{\prime}_1,A^ {\prime}_2)$ with $A ^ {\prime}_1=\{e_s\}$ and $A ^ {\prime}_2 = E \setminus
\{e_s\}$.} We first show that the allocation {$\mathbf{A} ^ {\prime}$} is $\frac{3}{2}$-PMMS. For
agent 1, due to Lemma \ref{obs::3.1}, $c_1(e_s)\leq \textnormal{MMS}_1(2, E)$ holds. As for agent
2, since $c_2(e_s)\geq c_1(e_s) > \frac{1}{4}$, we have {$c_2(A ^ {\prime}_2) = c_2(E)- c_2(e_s) <
\frac{3}{4} \leq \frac{3}{2}\textnormal{MMS}_2(2, E)$}. Thus, the allocation {$\mathbf{A} ^
{\prime}$} is $\frac{3}{2}$-PMMS. In the following, we first derive an upper bound for
$c_2(O_1\setminus \{e_s\}) - c_1(O_1\setminus \{e_s\})$, then based on the bound, we provide the
target upper bound for the price of fairness. Since $c_1(O_1) > \frac{3}{4}$ and $c_2(e_s) -
c_1(e_s ) > \frac{1}{8}$, we have $c_2(O_1\setminus \{e_s\}) - c_1(O_1\setminus \{e_s\}) = c_2(O_1)
-c_1(O_1) - (c_2(e_s) - c_1(e_s)) < \frac{1}{8}$, and then, the following holds,
\[
\textnormal{Price of } \frac{3}{2}\textnormal{-PMMS} \leq 1 + \frac{c_2(O_1\setminus \{e_s\})
		-c_1(O_1\setminus \{e_s\}) }{c_1(O_1) + c_2(O_2)}
\leq 1 + \frac{1}{8}\times\frac{4}{3} = \frac{7}{6},
\]
which completes the proof of the upper bound.

Regarding lower bound, consider an instance $I$ with two agents and a set $E = \{
e_1,e_2,e_3,e_4\}$ of four chores. The cost function for agent 1 is: $c_1(e_1) = \frac{3}{8},
c_1(e_2) = \frac{3}{8} + \epsilon, c_1(e_3) = \frac{1}{8}-\epsilon, c_1(e_4) = \frac{1}{8}$ where
$\epsilon> 0 $ takes arbitrarily small value. For agent 2, here cost function is: {$c_2(e_1) =
c_2(e_2) = \frac{1}{2}, c_2(e_3) = c_2(e_4) = 0 $}. It is not hard to verify that
$\textnormal{MMS}_i(2, E) = \frac{1}{2}$ for any $i =1, 2$. In the optimal allocation, the
assignment is; $e_1,e_2$ to agent 1 and $e_3,e_4$ to agent 2, resulting in $\textnormal{OPT}(I) =
\frac{3}{4} + \epsilon$. Observe that to satisfy $\frac{3}{2}$-PMMS, agent 1 {cannot} receive both
chores $e_1, e_2$, and accordingly, the minimum social cost of a $\frac{3}{2}$-PMMS allocation is
$\frac{7}{8}$ by assigning $e_1$ to agent 1 and the rest chores to agent 2. Based on this instance,
when $n=2$, the price of $\frac{3}{2}$-PMMS is at least $\frac{\frac{7}{8}}{\frac{6}{8} +
\epsilon}\rightarrow \frac{7}{6}$ as $\epsilon \rightarrow 0$.
\end{proof}

We remark that if we have an \emph{oracle} for the maximin share, then our constructive proof of
Proposition~\ref{PoFn=2::PMMS} can be transformed into an efficient algorithm for finding a
$3/2$-PMMS allocation {whose cost is at most $\frac{7}{6}$ times the optimal social cost}. Moving
to other fairness criteria, we have the following uniform bound.

\begin{proposition}\label{thm::price of PMMS and MMS}
When $n=2$ and agents have additive cost functions, the price of \textnormal{PMMS, MMS}, and
\textnormal{EFX} are all 2.
\end{proposition}

\begin{proof}
We first show results on the upper bound. When $n=2$, PMMS is identical with MMS and can imply EFX,
so it suffices to show that the price of PMMS is at most 2. Given an instance $I$, let allocation
$\mathbf{O} = (O_1,O_2)$ be its optimal allocation and w.l.o.g, we assume $c_1(O_1) \leq c_2(O_2)$.
If $c_2(O_2) \leq \frac{1}{2}$, then we have $c_1(O_1) \leq 1- c_1(O_1) = c_1(O_2)$ and $c_2(O_2)
\leq 1-c_2(O_2)  = c_2(O_1)$. So allocation $\mathbf{O}$ is an EF and accordingly is PMMS, which
yields the price of PMMS equals to one. Thus, we can further assume $c_2(O_2) >\frac{1}{2}$ and
hence the optimal social cost is larger than $\frac{1}{2}$.

We next show that there exist a PMMS allocation whose social cost is at most 1. W.l.o.g, we assume
$\textnormal{MMS}_1(2, E) \leq \textnormal{MMS}_2(2,E)$ (the other case is symmetric). Let
$\mathbf{T} = (T_1,T_2)$ be the allocation defining $\textnormal{MMS}_1(2, E)$ and $c_1(T_1) \leq
c_1(T_2) = \textnormal{MMS}_1(2, E)$. If $c_2(T_2) \leq c_2(T_1)$, then allocation $\mathbf{T}$ is
EF (also PMMS), and thus it hold that $c_1(T_1) \leq \frac{1}{2}$ and $ c_2(T_2) \leq \frac{1}{2}$.
Therefore, social cost of allocation $\mathbf{T}$ is no more than one, which implies the price of
PMMS is at most two. If $c_2(T_2) > c_2(T_1)$, then consider the allocation $\mathbf{T}^{\prime} =
(T_2, T_1)$. Since $c _ 1 ( T ^  {\prime} _ 1 ) = c_1(T_2) = \textnormal{MMS}_1(2, E)$ and $ c_2 (
T ^ {\prime} _ 2 ) = c_2(T_1) < c_2(T_2) $, then $\mathbf{T}^{\prime}$ is a PMMS allocation. Owing
to $\textnormal{MMS}_1(2, E) \leq \textnormal{MMS}_2(2,E)$, we claim that $c_2(T_1) \leq c_1(T_1)$;
otherwise, we have $\textnormal{MMS}_1(2, E) = c_1(T_2) > c_2(T_2) > c_2(T_1)$, and equivalently,
allocation $\mathbf{T^{\prime}}$ is a 2-partition where the cost of both bundles for agent 2 is
strictly smaller than $\textnormal{MMS}_1(2, E)$, contradicting to $\textnormal{MMS}_1(2, E) \leq
\textnormal{MMS}_2(2, E)$. By $c_2(T_1) \leq c_1(T_1)$, the social cost of allocation
$\mathbf{T}^{\prime}$ satisfies $c_2(T_1) + c_1(T_2) \leq 1$ and so the price of PMMS is at most
two.

Regarding the tightness, consider an instance $I$ with two agents and a set $E = \{ e_1,e_2,e_3\}$
of three chores. The cost function of agent 1 is : $c_1(e_1) = \frac{1}{2}, c_1(e_2) = \frac{1}{2}
- \epsilon$ and $c_1(e_3) = \epsilon$ where $\epsilon > 0 $ takes arbitrarily small value. For
agent 2, his cost is $c_2(e_1) = \frac{1}{2}, c_2(e_2) = \epsilon$ and $c_2(e_3) = \frac{1}{2} -
\epsilon$. An optimal allocation assigns chores $e_1, e_2$ to agent 2, and $e_3$ to agent 1, and
consequently, the optimal social cost equals to $\frac{1}{2} + 2\epsilon$. We first concern the
tightness on the notion of PMMS (or MMS, these two are identical when $n=2$). In any PMMS
allocations, it must be the case that an agent receives chore $e_1$ and the other one receives
chores $e_2,e_3$, and thus the social cost of PMMS allocations is one. Therefore, the price of PMMS
and of MMS is at least $\frac{1}{\frac{1}{2} + \epsilon} \rightarrow 2$ as $\epsilon \rightarrow 0
$. As for EFX, similarly, it must be the case that in any EFX allocations, the agent receiving
chore $e_1$ cannot receive any other chores. Thus, it not hard to verify that the social cost of
EFX allocations is also one and the price of EFX is at least $\frac{1}{\frac{1}{2} + \epsilon}
\rightarrow 2$ as $\epsilon \rightarrow 0 $.
\end{proof}

\subsection{More than two agents}

Note that the existence of an MMS allocation is not guaranteed in general
\cite{kurokawaFairEnoughGuaranteeing2018, azizAlgorithmsMaxminShare2017} and the existence of PMMS
or EFX allocation is still open in chores when $n\geq 3$. Nonetheless, we are still interested in the prices
of fairness in case such a fair allocation does exist.

\begin{proposition}\label{prop::7.3}
When agent have additive cost functions, for $n \geq 3$, the price of \textnormal{EF1, EFX, PMMS}
and $\frac{3}{2}$-\textnormal{PMMS} are all infinity.
\end{proposition}

\begin{proof}
In this proof, $\epsilon$ always takes arbitrarily small positive value. Based on our results on
the connections between fairness criteria, we have the relationship:
PMMS$\rightarrow$EFX$\rightarrow$EF1$\rightarrow$$\frac{3}{2}$-PMMS, where $A \rightarrow B $
refers to that notion $A$ is stricter than notion $B$. Therefore, it suffices to give a proof for $\frac{3}{2}$-PMMS.

Consider an instance with $n$ agents and $m\geq 5$ chores. The cost function of agent 1 is
$c_1(e_1) = 1-4\epsilon$, $c_1(e_j) = 0$ for $j=2,\ldots, m -4$, and $c_1(e_j) = \epsilon$ for $ j
\geq m-3$. For agent 2, his cost is $c_2(e_1) = 1-\frac{4}{m}$, $c_2(e_j) = 0$ for $j=2,\ldots, m
-4$, and $c_2(e_j) = \frac{1}{m}$ for $ j \geq m-3$. The cost function of agent 3 is: $c_3(e_1) =
\epsilon$, $c_3(e_j) = \frac{1}{m}$ for $ j =2,\ldots, m-1$, and $c_3(e_m)= \frac{1}{m} -
\epsilon$. For any $ i \geq 4$, the cost function of agent $i$ is $c_i(e _ j ) = \frac{1}{m}$ for
any $ j  \in [m]$. An optimal allocation assigns $e_{m-3}, e_{m-2}, e_{m-1},e_m$ to agent 1 and
$e_1$ to agent 3. For each of rest chore, it is assigned to the agents having zero cost on it.
Accordingly, the optimal social cost is $5\epsilon$. However, in any optimal allocation
$\mathbf{O}$, we have $\textnormal{MMS}_1(2, O_1 \cup O  _ 2  )  = 2\epsilon $, implying $ c _ 1  (
O _ 1 ) > \frac{3}{2}\textnormal{MMS}_1(2, O_1 \cup O  _ 2  ) $. Thus, agent 1 violates
$\frac{3}{2}$-PMMS. In order to achieve $\frac{3}{2}$-PMMS, at least one of $e_{m-3}, e_{m-2},
e_{m-1},e_m$ has to be assigned to someone other than agent 1, and so the social cost of a
$\frac{3}{2}$-PMMS allocation is at least $\frac{1}{m}  + 3\epsilon$,	resulting in an unbounded
price of $\frac{3}{2}$-PMMS when $\epsilon \rightarrow 0 $.
\end{proof}

In the context of goods allocation, \citet{barmanSettlingPriceFairness2020} present an
asymptotically tight price of EF1, $O(\sqrt{n})$. However, as shown by Proposition~\ref{prop::7.3},
when allocating chores, the price of EF1 is infinite, which shows a sharp contrast between goods
and chores allocation.

We are now left with MMS fairness. Let us first provide upper and lower bounds on the price of MMS.

\begin{proposition}\label{prop:my-add4}
When agents have additive cost functions, for $n \geq 3$, the price of \textnormal{MMS} is at most
$n^2$ and at least $\frac{n}{2}$.
\end{proposition}

\begin{proof}
We first prove the upper bound part. For any instance $I$, if $\textnormal{OPT}(I) \leq
\frac{1}{n}$, then by Lemma~\ref{obs::3.1}, any optimal allocations is MMS. Thus, we
can further assume $\textnormal{OPT}(I) > \frac{1}{n}$. Notice that the maximum social cost of an
allocation is $n$ and thus the upper bound of $n^2$ is straightforward.

For the lower bound, consider an instance $I$ with $n$ agents and $n+1$ chores $E=\{
e_1,\ldots,e_{n+1}\}$. For agent $i=2,\ldots,n$, $c_i(e_1) = c_i(e_2) = \frac{1}{2}$ and $c_i(e_j)
= 0$ for any $ j \geq 3$. As for agent 1, $c_1(e_1) = \frac{1}{n}$, $c_1(e_2) = \epsilon$, $
c_1(e_3) = \frac{1}{n} - \epsilon$ and $c_1(e_j )= \frac{1}{n}$ for any $ j \geq 4$ where $\epsilon
> 0$ takes arbitrarily small value. It is not hard to verify that $\textnormal{MMS}_1(n, E) =
\frac{1}{n}$ and $ \textnormal{MMS}_i(n, E) = \frac{1}{2}$ for $ i \geq 2$. In any optimal
allocation $\mathbf{O} = (O_1,\ldots,O_n)$, the first two chores are assigned to agent 1 and each
of the remaining chores is assigned to agents having cost zero. Thus, we have $\textnormal{OPT}(I)
= \frac{1}{n} + \epsilon$. However, in any optimal allocation $\mathbf{O}$, we have $c_1 ( O  _ 1 ) >
\textnormal{MMS}_1(n, E)  = \frac{1}{n}$. In order to achieve MMS, agent 1 can not receive both
chores $e_1,e_2$, and so at least one of them has to be assigned to the agent other than agent 1.
As a result, the social cost of an MMS allocation is at least $\frac{1}{2} + \epsilon$, which
implies that  the price of MMS is at least $\frac{n}{2}$ as $\epsilon\rightarrow 0 $.
\end{proof}

As mentioned earlier, the existence of MMS allocation is not guaranteed. So we also provide an
asymptotically tight price of 2-MMS, whose existence is guaranteed for any instance with additive
cost functions.

\begin{proposition}\label{prop:my-add5}
When agents have additive cost functions, for $n \geq 3$, the price of 2-\textnormal{MMS} is at
least $\frac{n+3}{6}$ and at most $n$, asymptotically tight $\Theta(n)$.
\end{proposition}

\begin{proof}
We first prove the upper bound. By Proposition~\ref{prop::4.6}, we know that an EF1 allocation is
also $\frac{2n-1}{n}$-MMS (also 2-MMS). As we mentioned earlier, round-robin algorithm always
output EF1 allocations. Consequently, given any instance $I$, the allocation returned by
round-robin is also 2-MMS. In the following, we incorporate the idea of expectation in probability
theory and show that there exists an order of round-robin such that the output allocation has
social cost at most 1.

Let $\sigma$ be a uniformly random permutation of $\{1,\ldots,n\}$ and $\mathbf{A}(\sigma) =
(A_1(\sigma),\ldots, A_n(\sigma))$ be the allocation returned by round-robin based on the order
$\sigma$. Clearly, each element $A_i(\sigma)$ is a random variable. Since $\sigma$ is chosen
uniformly random, the probability of agent $i$ on $j$-th position is $\frac{1}{n}$. Fix an agent
$i$, we assume $c_i(e_1) \leq c_i(e_2) \leq \cdots \leq c_i(e_m)$. If agent $i$ is in $j$-th
position of the order, then his cost is at most $c_i(e_j) + c_i(e_{n+j}) + \cdots + c_i(e_{\lfloor
\frac{m-j}{n} \rfloor n + j })$. Accordingly, his expected cost is at most $\sum_{j=1} ^{n}
\frac{1}{n} \sum_{l=0}^{\lfloor \frac{m-j}{n} \rfloor} c_i(e_{ln + j})$. Thus, we have an upper
bound of the expected social cost,
\[
\mathbb{E}[SC(\mathbf{A} (\sigma))]  \leq \sum_{i=1}^{n} \sum_{j=1} ^{n} \frac{1}{n}
\sum_{l=0}^{\lfloor \frac{m-j}{n} \rfloor} c_i(e_{ln + j})
= \frac{1}{n}\sum_{i=1}^{n} c_i(E) = 1.
\]
Therefore, there exists an order such that the social cost of the output is at most 1. Notice that
for any instance $I$, if $\textnormal{OPT}(I) \leq \frac{1}{n}$, then any optimal allocations are
also MMS. Thus, we can further assume $\textnormal{OPT}(I) >\frac{1}{n}$, and accordingly, the
price of 2-MMS is at most $n$.

For the lower bound, consider an instance $I$ with $n$ agents and a set $E=\{e_1,\ldots,e_{n+3}\}$
of $n+3$ chores. The cost function of agent 1 is:  $c_1(e_1) = c_1(e_2) = \frac{1}{n} - \epsilon$,
$c_1(e_3) = 3\epsilon$, $c_1(e_4) = c_1(e_5) = \epsilon$, $c_1(e_6) = \frac{1}{n} - 3\epsilon$
where $\epsilon > 0 $ takes arbitrarily small value, and $c_1(e_j) = \frac{1}{n}$ for any $ j> 6 $
(if exists). For agent $i=2,\ldots,n$, his cost is: $c_i(e_j) = \frac{1}{3}$ for any $j \in [3]$
and $c_i(e_j) = 0$ for $j \geq 4$. It is not hard to verify that $\textnormal{MMS}_1(n, E) =
\frac{1}{n}$ and $ \textnormal{MMS}_i(n, E) = \frac{1}{3}$ for any $ i \geq 2$. In any optimal
allocation $\mathbf{O} = ( O_1,\ldots,O_n)$, the first three chores are assigned to agent 1 and all
rest chores are assigned to agents having cost zero on them. Thus, we have $\textnormal{OPT}(I) =
\frac{2}{n}+\epsilon$. However, in any optimal allocations  $\mathbf{O}$, $    \frac{2}{n} +
\epsilon  = c _ 1 ( O _ 1 )  > 2\textnormal{MMS}_1(n, E)$ holds, and so agent 1 violates 2-MMS. In
order to achieve a 2-MMS allocation, agent 1 can not receive all first three chores, and so at
least one of them has to be assigned to the agent other than agent 1. As a result, the social cost
of an 2-MMS allocation is at least $\frac{1}{3} + \frac{1}{n} + 2\epsilon$, yielding that the price
of 2-MMS is at least $\frac{n}{6} + \frac{1}{2}$. Combing lower and upper bound, the price of 2-MMS
is $\Theta(n)$
\end{proof}


\section{Price of fairness beyond additive setting}

In this section, we study the price of fairness when agents have submodular cost functions. Notice
that for those fairness notions whose price of fairness is unbounded in the additive setting, the
efficiency loss would still be unbounded in the submodular setting. As a consequence, for most
notions, we only need to study its price of fairness in the case of two agents. Recall that, when
studying specific fairness notion, we only consider instances for which allocations satisfying the
underlying fairness notion do exist. All results established in this section remain true if agents
have subadditive cost functions.

\begin{proposition}\label{sub_EFX}
When $n=2$ and agents have submodular cost functions, if an \textnormal{EFX} allocation exists, the
price of \emph{EFX} is at least 3 and at most 4.
\end{proposition}

\begin{proof}
We first prove the upper bound. For an instance $I$, let $\mathbf{O} = ( O _ 1, O _ 2 )$ be an
optimal allocation, and w.l.o.g., we assume $ c _ 1 ( O _ 1 ) \leq c _ 2 (O_2)$. Since $ c _ 2 (\cdot)$ is
submodular and also subadditive, then $ c _  i  (  O _ i ) + c  _  i ( O _  { 3 - i }) \geq c _  i
( E )$ holds for $   i \in [2]$. If $  c _ 1 (  O _  1 )\leq c  _   2 ( O _ 2 )\leq 1/2$, then $  c
_ i  ( O _ {3-i} ) \geq  c _  i ( E  ) - c _ i ( O _ i   ) \geq 1/2 \geq  c _ i ( O _ i )$ holds
for $ i \in [2]$. Accordingly, allocation $\mathbf{O}$ is already EFX and we are done. Thus,
w.l.o.g., we can further assume $ c  _ 2  ( O _ 2 ) > 1/2$. Notice that the social cost of an
allocation is at most 2, and so the price of EFX is at most 4.

As for the lower bound, let us consider an instance with a set $E = \{ e _ 1, e _ 2, e _3\} $ of
three chores. The cost function of agent 1 is: $ c _ 1 ( e _ 1 ) = 1/2, c _ 1 ( e _ 2 ) = 1/2 -
\epsilon , c _ 1 ( e _ 3 ) = \epsilon$ and for any $S \subseteq E, c _ 1 ( S  ) = \sum _ { e \in S
} c _ 1 ( e _ s )$ where $\epsilon > 0 $ takes arbitrarily small value. The cost function of agent
2 is: $ c _ 2 ( e _ 1 ) = 1 - \epsilon, c _ 2 ( e _ 2 ) = 3\epsilon , c _ 2 ( e _ 3 ) = 1 -2
\epsilon$ and for any $S \subseteq E, c _ 2 ( S ) = \min\{ \sum_{e \in S} c _ 2 ( e) , 1\}$.
Function $ c _ 1 (\cdot)$ is additive and hence clearly monotone and submodular. For function $ c _
2 (\cdot)$, since $\sum_{e\in S} c_2(e)$ is additive (also monotone and submodular) on $S$, it follows that $c_2(\cdot)$ is also monotone and submodular (see Footnote~\thefootnote).

For this instance, the optimal allocation $\mathbf{O} = ( O _ 1,  O _ 2 )$ is $O_1 = \{ e _ 1, e _
3 \}$ and $ O _ 2 = \{ e _ 2\}$, yielding social cost $1/2+4\epsilon$. But due to $ c _ 1 ( O _ 1
\setminus \{ e _ 3 \} ) = 1/2 > 1/2 - \epsilon = c _ 1 ( O _ 2 )$, agent 1 violates EFX in
$\mathbf{O}$. In an EFX allocation, agent 2 can not receive the whole $E$ or $\{ e_ 1, e _ 3 \}$ or
$\{ e _ 1, e _ 2\}$. Thus, the EFX allocation with the smallest social cost is $A_1 = \{ e _ 2, e _
3 \}$ and $A _ 2 = \{ e _ 1\}$, yielding social cost $3/2 - \epsilon$. As a consequence, the price
of EFX is at least $\frac{3/2 - \epsilon}{1/2 + 4\epsilon} \rightarrow 3$ as $\epsilon \rightarrow
0$.
\end{proof}

\begin{proposition}\label{prop::pof_sub_EF1}
When $n=2$ and agents have submodular cost functions, if an \textnormal{EF1} allocation exists, the
price of \textnormal{EF1} is at least 2 and at most 4.
\end{proposition}

\begin{proof}
For the upper bound part, similar to the proof of Proposition~\ref{sub_EFX}, we can w.l.o.g assume
$c _ 1 ( O _ 1 ) \leq c _ 2 ( O _ 2 )$ and $ c _ 2 ( O _ 2 ) > 1/2$; otherwise, $\mathbf{O}$ is already
EF1. Notice that the social cost of an allocation is at most 2, and so the price of EF1 is at most
4.

As for the lower bound, let us consider an instance $I$ with a set $E = \{ e_1, e _ 2, e _ 3 \}$ of
three chores. The cost function of agent 1 is: $c _ 1 ( e _ 1 ) = 1/3 + \epsilon, c _ 1 ( e _ 2 ) =
1/3, c _ 1 ( e _ 3 ) = 1/3 - \epsilon$ and for any $ S \subseteq E , c _ 1 ( S ) = \sum_{ e \in S}
c _ 1 ( e )$ where $\epsilon > 0 $ takes arbitrarily small value. The cost function of agent 2 is:
$ c _ 2 ( e _ 1 ) = 1 - \epsilon, c _ 2 ( e _ 2 ) = 1 - \epsilon, c _ 2 ( e _ 3 ) = \epsilon$ and
for any $S \subseteq E, c _ 2 ( S ) = \min \{ \sum_{ e\in S} c _ 2 (e), 1 \}$. Function
$c_1(\cdot)$ is additive and clearly monotone and submodular. For function $ c _ 2 (\cdot)$, since $ \sum_{e\in S} c_2(e)$ is additive (also monotone and submodular) on $S$, it follows that $ c _ 2 (\cdot)$ is also monotone and submodular (see Footnote~\thefootnote).

For this instance, the optimal allocation $\mathbf{O} = ( O_1,  O _ 2 )$ is $ O _ 1 = \{ e _ 1, e _
2 \}$ and $ O _ 2 = \{ e _ 3 \}$, yielding social cost $2/3 + 2\epsilon$. But since $ \min_{ e \in
O _ 1 } c _ 1 ( O _ 1 \setminus \{ e \}) = 1/3 > 1/3 - \epsilon = c _ 1 ( O _ 2 )$, agent 1
violates EF1 under allocation $\mathbf{O}$. In an EF1 allocation, agent 2 can not receive all
chores and can not receive both $ e _ 1, e _ 2$, either. Thus, the EF1 allocation with minimal
social cost is $\mathbf{A} = (A _ 1, A _ 2 )$ with $ A _ 1 = \{ e _ 2 \}$ and $A  _ 2 =\{ e _ 1, e
_ 3 \}$, yielding cost $4/3$. As a consequence, the price of EF1 is at least $\frac{4/3}{3/2 +
2\epsilon} \rightarrow 2 $ as $\epsilon \rightarrow 0 $.
\end{proof}

\begin{proposition}\label{prop::pof_sub_MMS}
When $n=2$ and agents have submodular cost functions, if an \textnormal{PMMS} allocation exists,
the price of \emph{PMMS} is 3.
\end{proposition}

\begin{proof}
According to Lemma~\ref{obs::3.1}, $\textnormal{MMS} _ i (2, E) \geq 1/2$ holds for any $ i \in
[2]$. Given an instance $I$ and allocation $\mathbf{O}$ with minimal social cost, we can assume allocation $\mathbf{O}$ is not MMS and w.l.o.g, agent 2 violates the condition of MMS. Let
$\mathbf{A}$ be an MMS allocation. Due to $c _ 2 ( A _ 2 ) \leq \textnormal{MMS}_  2(2, E)  < c
_ 2 ( O _ 2 )$, we have
$$
\frac{c _ 1 ( A_ 1 ) + c _ 2 ( A _ 2 )}{c _ 1 ( O _ 1) + c _ 2 ( O _ 2 ) }
< \frac{c _ 1 ( A_ 1) +\textnormal{MMS}_ 2 (2, E)}{\textnormal{MMS}_2(2, E)} \leq 3,
$$
where the last inequality transition is because $c_1 (A_ 1 ) \leq 1$ and $\textnormal{MMS}_2 ( 2, E
)\geq 1/2$.

As for the lower bound, let us consider an instance $I$ with a set $E = \{ e_1, e _ 2, e _ 3 \}$ of
chores. The cost function of agent 1 is: $ c  _1 ( e _ 1 ) = 1/2$, $c _ 1 ( e _ 2 ) = 1/2 -
\epsilon$, $c _ 1 ( e _ 3 ) = \epsilon$ and for $S \subseteq E, c _ 1 (S) = \sum_{e \in S} c _ 1 (e
)$. The cost function of agent 2 is: $c _ 2 ( e _ 1 ) = 1 - 2\epsilon$, $c _ 2 ( e _ 2 ) = 10
\epsilon$, $c _ 2 ( e _ 3 ) = 1 - 3\epsilon$, $c _ 2 ( e_ 1 \cup e _ 2 ) = 1$, $c _ 2 (e _ 1 \cup e
_ 3 ) = 1$, $c _ 2 ( e _ 2 \cup e _ 3) = 1 - \epsilon$, $c  _2 (E) = 1$. Function $ c _ 1 (\cdot)$
is additive and hence monotone and submodular. It is not hard to verify $c_2(\cdot)$ is monotone.
Suppose $c_2(\cdot)$ is not submodular, and accordingly, there exists $ S \subsetneq T \subseteq E
$ and $ e \in E \setminus T$ such that $ c _ 2 ( T \cup \{ e \} ) - c_ 2 ( T ) > c _ 2 ( S \cup \{
e \} ) - c _ 2 ( S ) $. Since $ c _  2 (\cdot)$ is monotone, we have $ c _ 2 ( S \cup \{ e \} ) - c
_ 2 ( S )  \geq 0 $ implying $c _ 2 ( T \cup \{ e \} ) - c_ 2 ( T ) > 0 $. If $|T| = 2$, the only
possibility is $T = e _ 2 \cup e _  3$ and adding $e_1$ to $T$ has margin $\epsilon$. But for any
$S\subsetneq T $ the margin of adding $e_1$ to $S$ is larger than $\epsilon$, contradiction. If
$|T| = 1$, then $c _ 2 ( S \cup \{ e \} ) - c _ 2 ( S )  = c _ 2 (e)$ that is the largest margin of
adding item $e$ to a subset, contradiction. Thus, function $c_2(\cdot)$ is also submodular.

For this instance, partition $ \{ \{e _ 1\}, \{ e _ 2, e _ 3 \} \}$ defines $\textnormal{MMS}_1(2,
E) = 1/2$, and $ \{ \{ e_1\}, \{ e _ 2, e _ 3 \} \}$ defines $\textnormal{MMS}_2(2, E) =
1-\epsilon$. The minimal social cost allocation $\mathbf{O} = (O_1, O _ 2 )$ with $ O _ 1 = \{ e _
1, e _ 3\}$ and $ O _ 2 = \{ e _ 2 \}$, resulting in minimal social cost $1/2 + 11\epsilon$. But $
c _ 1 ( O _ 1 ) = 1/2 + \epsilon > \textnormal{MMS}_1(2, E)$, and thus $\mathbf{O}$ is not MMS.
Observe that in an MMS allocation, agent 2 can only receive either a single chore or $\{ e _ 2, e _
3 \}$. The MMS allocation with minimal social cost is $\mathbf{A}$ with $A _ 1 =\{ e _ 2, e_ 3 \}$
and $A _2 = \{ e _ 1\}$ whose social cost is equal to $3/2 - 2\epsilon$. As a consequence, the
price of MMS is at least $\frac{3/2 - 2\epsilon}{1/2 + 11\epsilon} \rightarrow 3 $ as $\epsilon
\rightarrow 0 $.
\end{proof}

\begin{proposition}\label{prop::pof_sub_1.5PMMS}
When $n=2$ and agents have submodular cost functions, if a $\frac{3}{2}$-$\textnormal{PMMS}$
allocation exists, the price of $\frac{3}{2}$-$\textnormal{PMMS}$ is at least $4/3$ and at most
$8/3$.
\end{proposition}

\begin{proof}
We first prove the upper bound. According to Lemma~\ref{obs::3.1}, $\textnormal{MMS}_i(2, E)  \geq
1/2$ holds for any $ i \in [2]$. Given an instance $I$, let $\mathbf{O} = ( O  _ 1,  O _ 2 )$ be an
minimal social cost allocation of $I$, and w.l.o.g., we assume $ c _ 1 ( O _ 1) \leq c _ 2 ( O _ 2
)$. Moreover, we can assume $ c _ 2 ( O _ 2 ) > 3/4$; otherwise $\mathbf{O}$ is already a
$\frac{3}{2}$-PMMS allocation and we are done. Notice that the social cost of an allocation is at
most 2, and so the price of $\frac{3}{2}$-PMMS is at most $8/3$.

As for the lower bound, let us consider an instance with a set $E = \{e _ 1, e _ 2, e _ 3, e _ 4
\}$ of four chores. The cost profile of agent 1 is: $ c _ 1 ( e _ 1 ) = 3/8$, $c _ 1 ( e _ 2) = 3/
8 + \epsilon$, $c _ 1 ( e _ 3) = 1/8 - \epsilon$, $c _ 1 ( e _ 4) = 1/8$ and for $ S \subseteq E, c
_ 1 ( S ) = \sum _ { e \in S} c _ 1 ( e)$. The cost profile of agent 2 is: $ c _ 2 ( e _ 1 ) = c _
2 ( e _ 2 ) = 1 - \epsilon$, $c _ 2 ( e _ 3 ) = c _ 2 ( e _ 4 ) = \epsilon$ and for $ S \subseteq
E, c _ 2 ( S ) = \min \{ \sum _ {e \in S} c _ 2 (e), 1\}$ where $\epsilon>0$ can take arbitrarily
small value. Function $c_1(\cdot)$ is additive and hence monotone and submodular. For function $c_2(\cdot)$, since $\sum_{e \in S} c _ 2 (e)$ is additive (also monotone and submodular) on $S$, it follows that $ c _2(\cdot)$ is also monotone and submodular (see Footnote~\thefootnote).

For the quantity of MMS, partition $\{\{ e _ 1, e _ 4\}, \{ e _ 2, e _ 3 \} \}$ defines
$\textnormal{MMS}_1(2, E) = 1/2$, and any allocation defines $\textnormal{MMS}_2(2, E) = 1$. The
minimal social cost allocation $\mathbf{O}$ with $ O _ 1 = \{ e _ 1, e _ 2\}$ and $ O _ 2 = \{ e _
3, e _ 4 \}$ whose social cost is equal to $3/4 + 3\epsilon$. But due to $  c _ 1 ( O _ 1 ) = 3/4 +
\epsilon > 3/2\cdot\textnormal{MMS}_1(2, E)$, agent 1 violates $\frac{3}{2}$-PMMS under
$\mathbf{O}$. Notice agent 1 can not receive both $e_1, e_2$, one can check that the
$\frac{3}{2}$-PMMS allocation with minimal social cost assigns all chores to agent 2, yielding
social cost exactly 1. As a consequence, the price of $\frac{3}{2}$-PMMS is at least $\frac{1}{3/4
+ 3\epsilon} \rightarrow \frac{4}{3}$ as $\epsilon \rightarrow 0 $.
\end{proof}

\begin{proposition}\label{prop::pof_sub_2MMS}
When $n=2$ and agents have submodular cost functions, the price of $2$-\textnormal{MMS} is 1.
\end{proposition}

\begin{proof}
According to Lemma~\ref{lemma::3.2}, the allocation with minimal social cost must also be 2-MMS,
completing the proof.
\end{proof}

\begin{proposition}\label{prop::pof_sub_2MMS_n}
When $ n \geq 3$ and agents have submodular cost functions, the price of 2-\textnormal{MMS} is at
least $\frac{n + 3}{6}$ and at most $\frac{n^2}{2}$.
\end{proposition}

\begin{proof}
The lower bound directly follows from the instance constructed in Proposition~\ref{prop:my-add5}.
As for the upper bound, given any minimal social cost allocation $\mathbf{O}$, if $\max_{ i \in
[n]} c _ i ( O _ i ) \leq \frac{2}{n}$, then due to $\textnormal{MMS}_i( n , E ) \geq \frac{1}{n}$
from Lemma~\ref{obs::3.1}, we have $ c _ i ( O _ i ) \leq 2\textnormal{MMS}_i(n ,E)$ for any $ i
\in [n]$. This implies allocation $\mathbf{O}$ is 2-MMS and we are done. Thus, we can assume
w.l.o.g.\ that $\max_{ i \in [n]} c _ i ( O _ i ) > \frac{2}{n}$. Notice the social cost of an
allocation is at most $n$, and so the price of 2-MMS is at most $\frac{n^2}{2}$.
\end{proof}

\section{Conclusions}

In this paper, we are concerned with fair allocations of indivisible chores among agents under the
setting of both additive and submodular (subadditive) cost functions. First, under the additive
setting, we have established pairwise connections between several (additive) relaxations of the
envy-free fairness in allocating, which look at how an allocation under one fairness criterion
provides an approximation guarantee for fairness under another criterion. Some of our results in
that part are in sharp contrast to what is known in allocating indivisible goods, reflecting the
difference between goods and chores allocation. We have also extended to the submodular setting and
investigated the connections between these fairness criteria. Our results have shown that, under
the submodular setting, the interesting connections we have established under the additive setting
almost disappear and few non-trivial approximation guarantees exist. Then we have studied the
trade-off between fairness and efficiency, for which we have established the price of fairness for
all these fairness notions in both additive and submodular settings. We hope our results have
provided an almost complete picture for the connections between these chores fairness criteria
together with their individual efficiencies relative to social optimum.

\bibliography{Connections_Efficiencies}

\clearpage
\newpage
\renewcommand{\thesection}{A} \setcounter{equation}{0} \renewcommand{\theequation}{A-\arabic{equation}}
\renewcommand{\thepage}{A} \setcounter{page}{1} \renewcommand{\thepage}{A-\arabic{page}}
\section*{Appendix}

\subsection{Proof of Proposition~\ref{prop:my-add1}}\label{proof of prop:my-add1}

We first prove the upper bound. Let $\mathbf{A}=(A_1,A_2,A_3)$ be a PMMS allocation and we focus on
agent 1. For the sake of contradiction, we assume $c_1(A_1) >\frac{4}{3}\textnormal{MMS}_1(3,E)$.
We can also assume bundles $A_1, A_2,A_3$ do not contain chore with zero cost for agent 1 since the
existence of such chores do not affect approximation ratio of allocation $A$ on PMMS or MMS. To
this end,  we let $c_1(A_2) \leq c_1(A_3)$ (the other case is symmetric).

We first show that $A_1$ must be the bundle yielding the largest cost for agent 1. Otherwise, if
$c_1(A_1) \leq c _1(A_2) \leq c_1(A_3)$, then by additivity $c_1(A_1) \leq \frac{1}{3} c _1(E) \leq
\textnormal{MMS}_1(3,E)$, contradicting to $c_1(A_1) >\frac{4}{3}\textnormal{MMS}_1(3,E)$. Or if
$c_1(A_2) < c_1(A_1)\leq c_1(A_3)$, since $A_1$ and $A_2$ is a 2-partition of $A_1\cup A_2$, then
$c_1(A_1)$ is at least $\textnormal{MMS}_1(2,A_1\cup A_2)$. On the other hand, since $\mathbf{A}$
is a PMMS allocation, we know $c_1(A_1) \leq \textnormal{MMS}_1(2,A_1\cup A_2)$, and thus,
$c_1(A_1)= \textnormal{MMS}_1(2,A_1\cup A_2)$ holds. Based on assumption $c_1(A_1)
>\frac{4}{3}\textnormal{MMS}_1(3,E)$ and Lemma \ref{obs::3.1}, we have $c_1(A_1) >
\frac{4}{3}\textnormal{MMS}_1(3,E) \geq \frac{4}{9}c_1(E)$, then $c_1(A_3) \geq c_1(A_1) >
\frac{4}{9}c_1(E)$ which yields $c_1(A_2) < \frac{1}{9}c_ 1(E)$ owning to the additivity. As a
result, the difference between $c_1(A_1)$ and $c_1(A_2)$ is lower bounded $c_1(A_1) - c_1(A_2) >
\frac{1}{3}c_1(E)$. Due to $c_1(A_1) = \textnormal{MMS}_1(2, A_1\cup A_2)$, we can claim that every
single chore in $A_1$ has cost strictly greater than $\frac{1}{3}c_1(E)$; otherwise, $\exists e\in
A_1$ with $c_1(e)\leq \frac{1}{3}c_1(E)$, then reassigning chore $e$ to $A_2$ yields a 2-partition
$\{ A_1\setminus \left\{e\right\}, A_2\cup \{ e\} \}$ with $\max \{c_1(A_1\setminus
\left\{e\right\}), c_1(A_2\cup \{ e\}) \} < c_1(A_1) = \textnormal{MMS}_1(2,A_1\cup A_2)$,
contradicting to the definition of maximin share. Since every single chore in $A_1$ has cost
strictly greater than $\frac{1}{3}c_1(E)$, then $A_1$ can only contain a single chore; otherwise,
$c_1(A_3) \geq c_1(A_1) \geq \frac{|A_1|}{3}c_1(E) \geq \frac{2}{3}c_1(E)$, implying $c_1(A_3 \cup
A_1) \geq\frac{4}{3}c_1(E)$, contradiction. However, if $|A_1| = 1$, according to the second point
of Lemma \ref{obs::3.1}, $c_1(A_1) >\frac{4}{3}\textnormal{MMS}_1(3,E)$ can never hold. Therefore,
it must hold that $c_1(A_1) \geq c _1(A_3) \geq c _1(A_2)$, which then implies $c_1(A_1) =
\textnormal{MMS}_1(2, A_1\cup A_3) = \textnormal{MMS}_1(2, A_1\cup A_2)$ as a consequence of PMMS.

Next, we prove our statement by carefully checking the possibilities of $|A_1|$. According to
Lemma~\ref{obs::3.1}, if $|A_1| = 1$, then $c_1(A_1) \leq \textnormal{MMS}_1(3, E)$. Thus, we can
further assume $|A_1| \geq 2$. We first consider the case $|A_1| \geq 3$. Since $c_1(A_1) >
\frac{4}{3}\textnormal{MMS}_1(3, E) \geq \frac{4}{9}c_1(E)$, by additivity, we have $c_1(A_2) +
c_1(A_3) < \frac{5}{9}c_1(E)$ and moreover, $c_1(A_2) < \frac{5}{18}c_1(E)$ due to $c_1(A_2)\leq
c_1(A_3)$. Then the cost difference between bundle $A_1$ and $A_2$ satisfies $c_1(A_1) - c_1(A_2) >
\frac{1}{6}c_1(E)$. This allow us to claim that every single chore in $A_1$ has cost strictly
greater than $\frac{1}{6}c_1(E)$; otherwise, reassigning a chore with cost no larger than
$\frac{1}{6}c_1(E)$ to $A_2$ yields another 2-partition of $A_1 \cup A_2$ in which the cost of
larger bundle is strictly smaller than $\textnormal{MMS}_1(2, A_1\cup A_2)$, contradiction. In
addition, since $c_1(A_1) = \textnormal{MMS}_1(2, A_1\cup A_2)$, we claim $c_1(A_2) \geq c_
1(A_1\setminus \{ e\}), \forall e \in A_1$; otherwise, $\exists e^{\prime} \in A_1$ such that
$c_1(A_2) < c_1(A_1\setminus \{ e^{\prime}\})$, then reassigning $e^{\prime}$ to $A_2$ yields
another 2-partition of $A_1 \cup A_2$ of which both two bundles' cost are strictly smaller than
$\textnormal{MMS}_1(2, A_1\cup A_2)$, contradiction. Thus, for any $e\in A_1$, we have $c_1(A_2)
\geq c_ 1(A_1\setminus \{ e\}) \geq\frac{1}{6}c_1(E)\cdot | A_1\setminus\{e\}|  \geq \frac{1}{3}
c_1(E)$, where the last transition is due to $|A_1| \geq 3$. However, the cost of bundle $A_2$ is
$c_1(A_2) < \frac{5}{18}c_1(E)$, contradiction.

The remaining work is to rule out the possibility of $|A_1| = 2$. Let $A_1 = \{ e ^{1} _1, e ^ {1}
_ 2\}$ with $c_1(e ^{1} _1) \leq c_ 1(e ^ {1} _ 2)$ (the other case is symmetric). Since $c_1(A_1)
> \frac{4}{3}\textnormal{MMS}_1(3,E) \geq \frac{4}{9}c_1(E)$, then $c_1(e^1_2) >
\frac{2}{9}c_1(E)$. Let $S^*_2 \in \arg\max_{S\subseteq A_2} \{ c_1(S): c_1(S)<c_1(e^1_1)\}$ (can
be empty set) be the largest subset of $A_2$ with cost strictly smaller than $c_1(e^1_1)$. Due to
$c_1(A_1) = \textnormal{MMS}_1(2, A_1\cup A_2)$, then swapping $S^*_2$ and $e^1_1$ would not
produce a 2-partition in which the cost of both bundles are strictly smaller than $c_1(A_1)$, and
thus $c_1(A_2\setminus S^*_2 \cup \{ e _1^1\}) \geq c_ 1(A_1)$, equivalent to
\begin{equation}\label{eq::17}
	c_1(A_2\setminus S^*_2) \geq c_ 1(e^1_2) > \frac{2}{9}c_1(E).
\end{equation}
Then, by $c_1(A_1) - c_1(A_2) > \frac{1}{6}c_1(E)$ and $c_1(A_2\setminus S^*_2) \geq c_ 1(e^1_2) $,
we have $c_1(e^1_1) - c_ 1(S _2^*) > \frac{1}{6} c _1(E)$, which allows us to claim that every
single chore in $A_2\setminus S_2^*$ has cost strictly greater than $\frac{1}{6}c_1(E)$; otherwise,
we can find another subset of $A_2$ whose cost is strictly smaller than $e_1^1$ but larger than
$c_1(S^*_2)$, contradicting to the definition of $S_2^*$. As a result, bundle $A_2\setminus S^*_2$
must contain a single chore; if not, $c_1(A_2) > \frac{1}{6}c_1(E) \cdot |A_2\setminus S^*_2|  \geq
\frac{1}{3}c_1(E)$, which implies $c_1(A_1 \cup A_2 \cup A_3) > \frac{10}{9}c_1(E)$ due to
$c_1(A_1) > \frac{4}{9}c_1(E)$ and $ c_1(A_3)\geq c _ 1(A_2) > \frac{1}{3}c_1(E)$. Thus, bundle
$A_2\setminus S^*_2$ only contains one chore, denoted by $e^2_1$. So we can decompose $A_2$ as $A_2
= \{e_1^2\} \cup S^*_2$ where $c_1(e^2_1) \geq c _1(e_2^1) > \frac{2}{9}c_1(E)$.

Next, we analyse the possible composition of bundle $A_3$. To have an explicit discussion, we
introduce two more notions $\Delta_1, \Delta_2$ as follows
\begin{equation}\label{eq::18}
	\begin{aligned}
		&c_1(A_1) = \frac{4}{9}c_1(E) + \Delta_1,\\
		&c_1(A_2)  = \frac{2}{9}c_1(E) + c _1(S^*_2) + \Delta_2.
	\end{aligned}
\end{equation}
Recall $c_1(A_1) > \frac{4}{9}c_1(E)$ and $c_1(e^2_1) \geq c_1(e^1_2)>\frac{2}{9}c_1(E)$, so both
$\Delta_1,\Delta_2 >0$. Similarly, let $S_3^*\in \arg\min_{S\subseteq A_3} \{c_1(S): c _1(S) <
c_1(e^1_1) \}$, then we claim $c_1(A_3\setminus S^*_3) \geq c_1(e^1_2)$; otherwise, swapping
$S_3^*$ and $e^1_1$ yields a 2-partition of $A_1\cup A_3$ in which the cost of both bundles are
strictly smaller than $c_1(A_1) = \textnormal{MMS}_1(2, A_1\cup A_3)$, contradicting to the
definition of maximin share. By additivity of cost functions and Equation (\ref{eq::18}), we have
$c_1(A_3) = \frac{3}{9} c_1(E) - c_1(S_2^*) - \Delta_1-\Delta_2$, and accordingly $c_1(A_1) -
c_1(A_3) = \frac{1}{9} c_1(E) + c_1(S_2^*) + 2\Delta_1+\Delta_2$. This combing $c_1(A_3\setminus
S^*_3) \geq c_1(e^1_2)$ yields
\begin{equation}\label{eq::19}
	c_1(e^1_1) - c_1(S^*_3) \geq \frac{1}{9}c_1(E) + c_1(S^*_2)+2\Delta_1+\Delta_2.
\end{equation}
Based on Inequality (\ref{eq::19}), we can claim that every single chore in $A_3\setminus S^*_3$
has cost at least $\frac{1}{9}c_1(E) + c_1(S^*_2)+2\Delta_1+\Delta_2$; otherwise, contradicting to
the definition of $S^*_3$. Recall $c_1(A_3)  = \frac{3}{9} c_1(E) - c_1(S_2^*) -
\Delta_1-\Delta_2$, then due to the constraint on the cost of single chore in $A_3\setminus S^*_3$,
we have $ |A_3\setminus S^*_3| \leq 2$. Meanwhile, $c_1(A_3\setminus S^*_3) \geq c_1(e^1_2)$
implying that bundle $A_3\setminus S^*_3$ can not be empty. In the following, we separate our proof
by discussing two possible cases: $|A_3\setminus S^*_3| = 1$ and $|A_3\setminus S^*_3| =2$.

\emph{Case 1:} $|A_3\setminus S^*_3| = 1$. Let $A_3\setminus S^*_3 = \{ e^3_1\}$. Therefore, the
whole set $E$ is composed by four single chores and two subsets $S^*_2, S^*_3$, i.e., $E =\{
e_1^1,e^1_2, e^2_1, S^*_2, e^3_1,S^*_3\}$. Then, we let $\mathbf{T} = (T_1,T_2,T_3)$ be the
allocation defining $\textnormal{MMS}_1(3,E)$ and without loss of generality, let $c_1(T_1)
=\textnormal{MMS}_1(3,E) $. Next, to find contradictions, we analyse bounds on both
$\textnormal{MMS}_1(3,E)$ and $c_1(A_1)$. Since $\min\{c_1(e^2_1), c_1(e^3_1) \} \geq c_1(e^1_2)
\geq \frac{1}{2}c_1(A_1)$, we claim that $c_1(A_1)\leq \frac{9}{18}c_1(E)$; otherwise $c_1(A_1) +
c_1(e^2_1)+c_1(e^3_1) >c_1(E)$. Notice that $E$ contains three chores with the cost at least
$\frac{2}{9}c_1(E)$ each, if any two of them are in the same bundle under $\mathbf{T}$, then
$\textnormal{MMS}_1(3,E) > \frac{4}{9}c_1(E)$ and consequently,$\frac{c_1(A_1)}
{\textnormal{MMS}_1(3,E)} < \frac{9}{8}$, contradiction. Or if each of $\{ e^1_2, e^2_1, e^3_1\}$
is contained in a distinct bundle, then the bundle also containing chore $e^1_1$ has cost at least
$c_1(A_1)$ as a result of $\min\{c_1(e^2_1), c_1(e^3_1) \} \geq c_1(e^1_2) $ and $A_1 = \{ e^1_1,
e^1_2\}$. Thus, $\textnormal{MMS}_1(3,E) \geq c_1(A_1)$ holds, contradicting to $c_1(A_1) >
\frac{4}{3}\textnormal{MMS}_1(3, E)$.

\emph{Case 2:} $|A_3\setminus S^*_3| = 2$. Let $A_3\setminus S^*_3 = \{e^3_1,e^3_2 \}$ and
accordingly, the whole set can be decomposed as $E=\{ e^1_1,e^1_2, e^2_1, S^*_2, e^3_1,e^3_2,
S^*_3\}$. Note the upper bound $c_1(A_1) \leq \frac{9}{18}c_1(E)$ still holds since $\min \{
c_1(A_3\setminus S^*_3), c_1(e^2_1)\} \geq c_ 1(e^1_2)$. Then, we analyse the possible lower bound
of $\textnormal{MMS}_1(3, E)$. If chores $e^1_2, e^2_1$ are in the same bundle of $\mathbf{T}$,
then $\textnormal{MMS}_1(3, E) > \frac{4}{9}c_1(E)$ holds and so
$\frac{c_1(A_1)}{\textnormal{MMS}_1(3, E)} < \frac{9}{8}$, contradiction. Thus, chores $e^1_2,
e^2_1$ are in different bundles in $\mathbf{T}$. Then, if both chores $ e^3_1, e^3_2$ are in the
bundle containing $e^1_2$ or $ e^2_1$, then we also have $\textnormal{MMS}_1(3, E) >
\frac{4}{9}c_1(E)$ implying $\frac{c_1(A_1)}{\textnormal{MMS}_1(3, E)} < \frac{9}{8}$,
contradiction. Therefore, only two possible cases; that is, both $e^3_1,e^3_2$ are in the bundle
different from that containing $e^1_2$ or $e^2_1$; or the bundle having $e^1_2$ or $e^2_1$ contains
at most one of $e^3_1,e^3_2$.

\emph{Subcase 1:} both $e^3_1,e^3_2$ are in the bundle different from that containing $e^1_2$ or
$e^2_1$; Recall $c_1(e^1_1) > \frac{3}{18}c_1(E)+c_1(S^*_2)$ and the fact $\min \{ c _1 (e^1_2),
c_1(e^2_1), c_1(e^3_1\cup e^3_2)\} > \frac{4}{18}c_1(E)$, the bundle also containing $e^1_1$ has
cost strictly greater than $\frac{7}{18} c _1(E)$. Thus, $\textnormal{MMS}_1(3, E) > \frac{7}{18} c
_1(E)$, which combines $c_1(A_1) \leq \frac{9}{18}c_1(E)$ implying
$\frac{c_1(A_1)}{\textnormal{MMS}_1(3, E)} < \frac{9}{7} < \frac{4}{3}$, contradiction.

\emph{Subcase 2:} bundle having $e^1_2$ or $e^2_1$ contains at most one of $e^3_1,e^3_2$. Recall
$c_1(e^2_1) \geq c_1(e^1_2)$ and $\min \{c_1(e^3_1) , c_1(e^3_2)\} \geq \frac{1}{9}c_1(E) +
c_1(S^*_2) + 2\Delta_1 + \Delta_2 $, thus in allocation $\mathbf{T}$ there always exist a bundle
with cost at least $\frac{1}{9}c_1(E) + c_1(S^*_2) + 2\Delta_1 + \Delta_2 + c_ 1(e^1_2)$ and
results in the ratio
\begin{equation}\label{eq::20}
\frac{c_1(A_1)}{\textnormal{MMS}_1(3,E)} \leq \frac{c_1(e^1_1) + c _1(e^1_2)}{\frac{1}{9}c_1(E)
+ c_1(S^*_2) + 2\Delta_1 + \Delta_2 + c_ 1(e^1_2)}.
\end{equation}
In order to satisfying our assumption $\frac{c_1(A_1)}{\textnormal{MMS}_1(3,E)} > \frac{4}{3}$, the
RHS of Inequality (\ref{eq::20}) must be strictly greater than $\frac{4}{3}$, which implies the
following
\begin{equation}\label{eq::21}
	c_1(e^1_1) > \frac{2}{9} c_ 1(E) + 2 c_1(S^*_1) + 4\Delta_1 + 2\Delta_2.
\end{equation}
However, based on the first equation of (\ref{eq::18}) and $c_1(e^1_1) \leq c _1(e^1_2)$, we have
$c_1(e^1_1) \leq \frac{2}{9}c _1(E) + \frac{1}{2}\Delta_1< \frac{2}{9} c_ 1(E) + 2 c_1(S^*_1) +
4\Delta_1 + 2\Delta_2$ due to $\Delta_1, \Delta_2 > 0 $. This contradicts to Inequality
(\ref{eq::21}). Therefore, $\frac{c_1(A_1)}{\textnormal{MMS}_1(3,E)} > \frac{4}{3}$ can never hold
under \emph{Case 2}. Up to here, we complete the proof of the upper bound.

Next, as for tightness, consider an instance with three agents and a set $E = \{ e _1,..., e _ 6\}$
of six chores. Agents have identical cost functions. The cost function of agent 1 is as follows:
$c_1(e_j) = 2, \forall j =1,2,3$ and $c_1(e_j) = 1, \forall j = 4,5,6$. It is easy to see that
$\textnormal{MMS}_1(3, E) = 3$. Then, consider an allocation $\mathbf{B} = \{ B_1,B_2,B_3\}$ with
$B_1=\{ e_1, e_2\}, B_2=\{ e_3\}$ and $B_3=\{ e_4,e_5,e_6\}$. It is not hard to verify that
allocation $\mathbf{B}$ is PMMS and due to $c_1(B_1) = 4$, we have the ratio $\frac{c
	_1(B_1)}{\textnormal{MMS}_1(3, E)} = \frac{4}{3}$.


\subsection{Algorithm 1}
\label{sec:EF1-algorithm}

{The following efficient algorithm, which we call $ALG_1$, outputs an EF1 allocation with a cost at
most $\frac{5}{4}$ times the optimal social cost under the case of $n=2$. In the algorithm, we use
notations:} $L(k):=\{ e_1,\ldots,e_k\}$ and $R(k):=\{ e_{k},\ldots,e_{m}\}$ for any $1 \leq k \leq
m$.

\begin{algorithm}[H]
\caption{\hspace{-2pt}{}}
\label{alg:social_optimal}
\begin{algorithmic}[1]
	\REQUIRE An instance $I$ with two agents.
	\ENSURE an EF1 allocation of instance $I$.
	\STATE Partition $E = E _0 \cup E _ 1 \cup E _ 2$ where $ E _ 1 = \{ e \in E \mid c_ 1(e) < c_2(e) \}$ and $E_2 = \{ e\in E \mid c _1(e) > c _ 2 (e) \}$ (we assume $c_1(E_1) \leq c_2(E_2)$ and the other case is symmetric).
	\STATE Order chores such that $\frac{c_1(e_1)}{c_2(e_1)} \leq \frac{c_1(e_2)}{c_2(e_2)}
	\leq \cdots \leq \frac{c_1(e_m)}{c_2(e_m)}$, tie breaks arbitrarily. For chore $e$
	with $c_1(e)=0$, put it at the front and chore $e$ with $c_2(e) = 0$ at back.
	\STATE Find index $s$ such that $c_1(e_s) < c_2(e_s)$ and $c_1(e_{s+1}) \geq c_2(e_{s+1})$.
	\IF{$s=0$}
	\STATE Run a round-robin algorithm: let each of the agent $1, \ldots, n$ picks her most preferred item in that order, and repeat until all chores are assigned.
	\RETURN the output
	\ELSE
	\STATE Let $\mathbf{O}$ be the allocation with $O_1=L(s)$ and $O_2 = R(s+1)$.
	\IF{allocation $\mathbf{O}$ is EF1}
	\RETURN allocation $\mathbf{O}$.
	\ELSE
	\STATE  find the maximum index $f \geq s$ such that $c_2(R(f+2)) > c_2(L(f))$.
	\RETURN allocation $\mathbf{A}$ with $A_1=L(f+1)$ and $A_2=R(f+2)$
	\ENDIF
	\ENDIF
\end{algorithmic}
\end{algorithm}

\end{document}